\PassOptionsToPackage{unicode}{hyperref}
\PassOptionsToPackage{hyphens}{url}
\PassOptionsToPackage{dvipsnames,svgnames,x11names}{xcolor}

\documentclass[12pt]{article}
\usepackage{amsmath,amssymb}
\usepackage{iftex}
\ifPDFTeX
  \usepackage[T1]{fontenc}
  \usepackage[utf8]{inputenc}
  \usepackage{textcomp} 
\else 
  \usepackage{unicode-math}
  \defaultfontfeatures{Scale=MatchLowercase}
  \defaultfontfeatures[\rmfamily]{Ligatures=TeX,Scale=1}
\fi
\usepackage{lmodern}
\ifPDFTeX\else  
\fi
\IfFileExists{upquote.sty}{\usepackage{upquote}}{}
\IfFileExists{microtype.sty}{
  \usepackage[]{microtype}
  \UseMicrotypeSet[protrusion]{basicmath} 
}{}
\makeatletter
\@ifundefined{KOMAClassName}{
  \IfFileExists{parskip.sty}{%
    \usepackage{parskip}
  }{
    \setlength{\parindent}{0pt}
    \setlength{\parskip}{6pt plus 2pt minus 1pt}}
}{
  \KOMAoptions{parskip=half}}
\makeatother
\usepackage{xcolor}
\makeatletter
\ifx\paragraph\undefined\else
  \let\oldparagraph\paragraph
  \renewcommand{\paragraph}{
    \@ifstar
      \xxxParagraphStar
      \xxxParagraphNoStar
  }
  \newcommand{\xxxParagraphStar}[1]{\oldparagraph*{#1}\mbox{}}
  \newcommand{\xxxParagraphNoStar}[1]{\oldparagraph{#1}\mbox{}}
\fi
\ifx\subparagraph\undefined\else
  \let\oldsubparagraph\subparagraph
  \renewcommand{\subparagraph}{
    \@ifstar
      \xxxSubParagraphStar
      \xxxSubParagraphNoStar
  }
  \newcommand{\xxxSubParagraphStar}[1]{\oldsubparagraph*{#1}\mbox{}}
  \newcommand{\xxxSubParagraphNoStar}[1]{\oldsubparagraph{#1}\mbox{}}
\fi
\makeatother

\usepackage{mathtools}
\usepackage{longtable,booktabs,array}
\usepackage{calc} 
\usepackage{etoolbox}
\makeatletter
\patchcmd\longtable{\par}{\if@noskipsec\mbox{}\fi\par}{}{}
\makeatother
\IfFileExists{footnotehyper.sty}{\usepackage{footnotehyper}}{\usepackage{footnote}}
\makesavenoteenv{longtable}
\usepackage{graphicx}
\makeatletter
\def\maxwidth{\ifdim\Gin@nat@width>\linewidth\linewidth\else\Gin@nat@width\fi}
\def\maxheight{\ifdim\Gin@nat@height>\textheight\textheight\else\Gin@nat@height\fi}
\makeatother
\setkeys{Gin}{width=\maxwidth,height=\maxheight,keepaspectratio}
\makeatletter
\def\fps@figure{htbp}
\makeatother

\makeatletter
\@ifpackageloaded{caption}{}{\usepackage{caption}}
\AtBeginDocument{%
\ifdefined\contentsname
  \renewcommand*\contentsname{Table of contents}
\else
  \newcommand\contentsname{Table of contents}
\fi
\ifdefined\listfigurename
  \renewcommand*\listfigurename{List of Figures}
\else
  \newcommand\listfigurename{List of Figures}
\fi
\ifdefined\listtablename
  \renewcommand*\listtablename{List of Tables}
\else
  \newcommand\listtablename{List of Tables}
\fi
\ifdefined\figurename
  \renewcommand*\figurename{Figure}
\else
  \newcommand\figurename{Figure}
\fi
\ifdefined\tablename
  \renewcommand*\tablename{Table}
\else
  \newcommand\tablename{Table}
\fi
}
\@ifpackageloaded{float}{}{\usepackage{float}}
\floatstyle{ruled}
\@ifundefined{c@chapter}{\newfloat{codelisting}{h}{lop}}{\newfloat{codelisting}{h}{lop}[chapter]}
\floatname{codelisting}{Listing}

\makeatother
\makeatletter
\@ifpackageloaded{caption}{}{\usepackage{caption}}
\@ifpackageloaded{subcaption}{}{\usepackage{subcaption}}
\makeatother

\ifLuaTeX
  \usepackage{selnolig}  
\fi
\usepackage[]{natbib}
\usepackage{bookmark}

\IfFileExists{xurl.sty}{\usepackage{xurl}}{} 
\hypersetup{
  pdftitle={Title},
  pdfauthor={Author 1; Author 2},
  pdfkeywords={3 to 6 keywords, that do not appear in the title},
  colorlinks=true,
  linkcolor={blue},
  filecolor={Maroon},
  citecolor={Blue},
  urlcolor={Blue},
  pdfcreator={LaTeX via pandoc}}

\usepackage{setspace}
\usepackage[T1]{fontenc}
\usepackage{float}
\usepackage{mathrsfs}
\usepackage{amsmath}
\usepackage{amssymb}
\usepackage{amsthm}

\usepackage{verbatim}
\usepackage{bbm}
\usepackage[color=yellow]{todonotes}
\hypersetup{
    colorlinks=true,
    linkcolor=blue,
    filecolor=magenta,      
    urlcolor=cyan,
    citecolor = blue,
}

\usepackage{tocloft}
\usepackage{titletoc}

\usepackage{xurl} 
\hypersetup{
    colorlinks=true,
    linkcolor=blue,
    filecolor=magenta,      
    urlcolor=cyan,
}
\usepackage{yhmath}
\makeatletter
\usepackage{geometry}
\usepackage{bm}
\numberwithin{equation}{section}

\usepackage{algorithm}
 \usepackage{algorithmicx}
\floatstyle{ruled}
\newfloat{algorithm}{tbp}{loa}
\providecommand{\algorithmname}{Algorithm}
\floatname{algorithm}{\protect\algorithmname}

\usepackage{enumerate}

\@ifundefined{definecolor}{\@ifundefined{definecolor}
 {\@ifundefined{definecolor}
 {\usepackage{color}}{}
}{}
}{}

\usepackage[all]{xy}

\newtheorem{theorem}{Theorem}[section]
\newtheorem{lem}{Lemma}[section]

\newcounter{hypA}
\newenvironment{hypA}{\refstepcounter{hypA}\begin{itemize}
  \item[({\bf A\arabic{hypA}})]}{\end{itemize}}
\newcounter{hypB}

\newcounter{hypD}

\usepackage{babel}\date{}

\def\l{\lambda}

\allowdisplaybreaks

\makeatother

\begin{document}

\begin{center}

{\Large \textbf{Bayesian Inference for Partially Observed McKean-Vlasov SDEs with Full Distribution Dependence}}

\vspace{0.5cm}

Ning Ning$^{a,*}$ and Amin Wu$^{b,*}$

{\footnotesize $^{a}$Department of Statistics,  Texas A\&M University,
College Station, TX, US.}\\
{\footnotesize $^{b}$Statistics Program, Computer, Electrical and Mathematical Sciences and Engineering Division, \\ King Abdullah University of Science and Technology, Thuwal, 23955-6900, KSA.} \\
{\footnotesize E-Mail:\,} \texttt{\emph{\footnotesize patning@tamu.edu}}, 
\texttt{\emph{\footnotesize amin.wu@kaust.edu.sa}}

\begingroup
    \renewcommand{\thefootnote}{*}
    \footnotetext{The authors are listed in alphabetical order.}
\endgroup

\begin{abstract}
McKean–Vlasov stochastic differential equations (MVSDEs) describe systems whose dynamics depend on both individual states and the population distribution, and they arise widely in neuroscience, finance, and epidemiology. In many applications the system is only partially observed, making inference very challenging when both drift and diffusion coefficients depend on the evolving empirical law. This paper develops a Bayesian framework for latent state inference and parameter estimation in such partially observed MVSDEs. We combine time-discretization with particle-based approximations to construct tractable likelihood estimators, and we design two particle Markov chain Monte Carlo (PMCMC) algorithms: a single-level PMCMC method and a multilevel PMCMC (MLPMCMC) method that couples particle systems across discretization levels. The multilevel construction yields correlated likelihood estimates and achieves mean square error $(O(\varepsilon^2))$ at computational cost $(O(\varepsilon^{-6}))$, improving on the $(O(\varepsilon^{-7}))$ complexity of single-level schemes. We address the fully law-dependent diffusion setting which is the most general formulation of MVSDEs, and provide theoretical guarantees under standard regularity assumptions. Numerical experiments confirm the efficiency and accuracy of the proposed methodology. 
\bigskip

\noindent \textbf{Keywords}: Bayesian estimation, Generalized Mckean-Vlasov stochastic differential equations, Markov Chain Monte Carlo, Multilevel Monte Carlo. 
\end{abstract}

\end{center}


\section{Introduction}
We first provide the background and motivation in Section~\ref{sec:background}, summarize our main contributions in Section~\ref{sec:Contributions}, followed by the overall organization of the paper in Section~\ref{sec:Organization}.

\subsection{Background and Motivation}
\label{sec:background}

McKean–Vlasov stochastic differential equations (MVSDEs) provide a fundamental framework for describing systems of interacting particles or agents whose dynamics depend on both their individual states and the overall distribution of the population. Such equations naturally arise in diverse fields, including mean-field games, statistical physics, neuroscience, and quantitative finance, where collective effects play a critical role in shaping individual behavior \citep{bald,biol,fin}. Despite their expressiveness, simulating MVSDEs poses significant challenges. The dependence on the law of the process introduces an intrinsic coupling between all particles, leading to high computational complexity and propagation-of-chaos errors in particle-based approximations \citep{basic_method,ning2026well}. Moreover, the nonlinear dependence on the empirical measure often precludes closed-form solutions and exacerbates numerical instability, particularly when the coefficients are non-globally Lipschitz \citep{ning2024one}. As a result, the accurate and efficient simulation of MVSDEs remains an active and technically demanding area of research.

While substantial progress has been made in the analysis and simulation of MVSDEs under fully observed settings \citep{nik, wen2016maximum}, many real-world systems are only partially observed. In such scenarios, the latent state evolves continuously in time, but only discrete, noisy measurements are available, providing an incomplete and indirect view of the underlying dynamics \citep{heng2024unbiased}. This partial-observation structure arises naturally in fields such as neuroscience, finance, and epidemiology, where the full system state cannot be measured directly or is prohibitively costly to obtain \citep{li2024inference, gu2025optimizing}. Modeling MVSDEs under partial information introduces significant challenges: the dependence of the dynamics on the population distribution must be inferred through observations, rendering parameter estimation a nonlinear filtering problem on the space of probability measures. These difficulties are further amplified when the diffusion term depends on the law of the process, since inference must simultaneously track both stochasticity and evolving distributional interactions, complicating both theoretical analysis and numerical computation.

Despite these complications, incorporating law-dependence in the diffusion coefficient is essential for accurately representing systems in which volatility itself exhibits mean-field effects. In neuroscience, for example, synaptic noise intensities in large neural populations depend on collective activity: aggregate firing patterns modulate membrane conductances and background fluctuations \citep{dos2022simulation}. Models with law-independent diffusion therefore miss key population-level volatility phenomena underlying noise-induced synchronization and stochastic resonance. In finance, similar effects arise in systemic risk modeling, where institutions’ idiosyncratic volatility depends on the evolving cross-sectional distribution of capital or exposure. For instance, \cite{fin} develop an asymptotic inference framework for large interacting systems with jump diffusions whose drift and diffusion coefficients depend on the empirical distribution, capturing correlated default, interbank dependence, and other distribution-driven volatility mechanisms. Across these applications, law-dependent diffusion is indispensable for describing comovement, volatility clustering, and distribution-mediated interactions that cannot be reproduced by simpler specifications.

The inclusion of law-dependent diffusion, however, greatly intensifies both the statistical and computational burden. The posterior distribution becomes highly intractable due to the intertwined roles of latent trajectories, unknown parameters, and evolving probability measures. Classical inference methods are insufficient: closed-form transition densities are unavailable \citep{beskos2006exact}, and numerical schemes must handle high-dimensional integrations over spaces of probability measures. These challenges prompt two natural questions: 
\begin{itemize}
\item Can one develop a broadly applicable, data-driven methodology for such general models?
\item Once such an approach is established, can its performance be further enhanced by reducing estimation variability?
\end{itemize}  
Addressing these questions is the central aim of this paper.

\subsection{Our Contributions}
\label{sec:Contributions}

This work develops a general Bayesian methodology for latent state inference and parameter estimation in partially observed MVSDE (PO-MVSDEs) with law-dependent drift and diffusion coefficients. We begin by combining appropriate time-discretization schemes with particle-based approximations \citep{po_mv} to obtain tractable likelihood estimators for models in which both components of the dynamics depend on the evolving empirical distribution. Building on these approximations, we design Markov chain Monte Carlo (MCMC) algorithms \citep{andrieu} tailored for posterior sampling in this setting. Since fine discretization levels incur substantial computational cost, we further incorporate multilevel Monte Carlo (MLMC) techniques \citep{giles, giles1, jasra_bpe_sde} to balance variance and bias across resolutions. The multilevel construction yields significant gains in efficiency over classical single-level schemes, and the accompanying theoretical analysis provides rigorous guarantees on approximation error and overall complexity.

We develop two particle MCMC (PMCMC) methods for Bayesian inference in discretized PO-MVSDEs. Particle filters, or sequential Monte Carlo methods, are a workhorse for non-linear stochastic filtering \citep{doucet2001sequential, chopin2020introduction}. 
First, we construct a single-level PMCMC algorithm based on a particle filter built from the approximated transition law $\mu_{k \Delta_l, \theta}^l$.  
The particle filter provides unbiased marginal likelihood estimates, which are embedded within a pseudo-marginal MCMC scheme to yield exact inference for the discretized model.
Second, to improve efficiency, we introduce a multilevel PMCMC (MLPMCMC) algorithm that couples the approximated laws at two consecutive levels $(l-1,l)$.  
This coupling enables the construction of correlated likelihood estimates through coupled particle filters, which are aggregated via weighted differences or averages to form a telescoping estimator with level-dependent variance decay.  
Throughout the hierarchy, the particle numbers $N_l$ and parameters $\theta$ remain fixed within each level, while the inter-level coupling is designed to minimize the variance of the multilevel decomposition.  
These methods extend the foundational PMCMC results of \cite{andrieu} to mean-field systems with law-dependent noise, and we show that exactness is retained even when only biased law approximations are available.

Our methodology generalizes the framework of \cite{jasra2025bayesian}, which treated partially observed McKean--Vlasov models with law-dependent drift and state-only diffusion.  
Their setting is recovered in our framework as the special case $\zeta_{2,\theta} \equiv 0$.  
Allowing $\zeta_{2,\theta} \neq 0$ introduces substantial new challenges: both interaction kernels $\zeta_{1,\theta}$ (drift) and $\zeta_{2,\theta}$ (diffusion) must be approximated; the law-dependent diffusion affects the stochastic component of the dynamics, complicating both the construction of coupled processes and the variance analysis; and the regularity conditions must hold uniformly over the measure space $\mathcal{P}(\mathbb{R}^d)$ rather than only with respect to state variables.  
Despite these challenges, we show that multilevel variance reduction remains effective, achieving computational cost $O(\varepsilon^{-6})$ for mean square error $O(\varepsilon^{2})$, which matches the optimal complexity in \cite{jasra2025bayesian} while handling a strictly broader model class.

The proof proceeds by decomposing the mean square error into several components, corresponding to different sources of approximation. Specifically, we distinguish between the error arising from the base-level MCMC sampler, the additional error introduced by the bi-level MCMC updates, the particle approximation error at the base discretization level, the coupled particle approximation errors across successive levels, and the bias induced by time discretization. The main technical challenge stems from the law-dependent diffusion term, whose presence influences the variance bounds associated with both interaction kernels, $\zeta_{1,\theta}$ and $\zeta_{2,\theta}$. To address this, we develop new variance estimates for the coupled system and establish that, under the regularity conditions imposed in Assumption~A\ref{ass:1}, these contributions can be controlled sufficiently to ensure that the overall algorithm attains the optimal multilevel convergence rate.

\subsection{Organization of the Paper}
\label{sec:Organization}

The rest of the paper proceeds as follows. Section~\ref{sec:model} introduces the overall model formulation, beginning with the PO-MVSDE framework, followed by its time-discretized version. Section~\ref{sec:method} develops the computational methodology for approximate inference under the discretized model, proposing two PMCMC approaches for posterior sampling over the PO-MVSDE. Section~\ref{sec:theory} presents the theoretical results of the proposed methodology, including the required regularity assumptions and the main convergence results. Section~\ref{sec:numerics} applies the proposed methods to a neuronal activity model governed by a multidimensional MVSDE with a non-globally Lipschitz structure, and presents the corresponding numerical results. Finally, Section \ref{sec:Conclusion} concludes the paper by summarizing the key findings, discussing broader implications, and suggesting directions for future research. Proofs of all theoretical results are provided in the Supplementary Material.

\section{Model Formulation}
\label{sec:model}
In this section, we introduce the overall model formulation, beginning with the PO-MVSDE framework in Section~\ref{sec:POMVD}, followed by its time-discretized setting in Section~\ref{sec:Time_Discretization}.

\subsection{Partially Observed McKean-Vlasov SDEs}
\label{sec:POMVD}
Consider a complete probability space $(\Omega, \mathscr{F}, \mathbb{F}=\{\mathscr{F}_t\}_{t\in [0,T]}, P)$ which supports an $\mathbb{F}$-adapted standard Brownian motion $W=\{W_t\}_{t\in[0,T]}$.
We consider a mean-field SDE of McKean-Vlasov type. Let $X_0 = x_0 \in \mathbb{R}^d$ be the initial condition and $\theta \in \Theta \subseteq \mathbb{R}^{d_{\theta}}$ be a fixed parameter. The process $\{X_t\}_{t \geq 0}$ satisfies the following SDE:
\begin{equation}\label{eq:sde}
dX_t = a_{\theta}\left(X_t,\overline{\zeta}_{1,\theta}(X_t,\mu_{t,\theta})\right)dt + b_{\theta}\left(X_t,\overline{\zeta}_{2,\theta}(X_t,\mu_{t,\theta})\right)dW_t,
\end{equation}
where the interaction terms are defined by
$$
\overline{\zeta}_{j,\theta}(X_t,\mu_{t,\theta}) = \int_{\mathbb{R}^d}\zeta_{j,\theta}(X_t,x)\mu_{t,\theta}(dx), \qquad j \in \{1,2\}.
$$
Here, $\mu_{t,\theta}$ denotes the law of $X_t$ under parameter $\theta$, with initial condition $\mu_{0,\theta}(dx) = \delta_{x_0}(dx)$, which is the Dirac measure concentrated at the fixed starting point $x_0$. For each $\theta \in \Theta$, the functions $\zeta_{j,\theta}: \mathbb{R}^{2d} \rightarrow \mathbb{R}$ represent interaction kernels, $a_{\theta}: \mathbb{R}^d \times \mathbb{R} \rightarrow \mathbb{R}^d$ is the drift coefficient function, and $b_{\theta}: \mathbb{R}^d \times \mathbb{R} \rightarrow \mathbb{R}^{d \times d}$ is the diffusion coefficient function.

We consider a discrete-time observation process $Y_1, Y_2, \ldots, Y_T$ taking values in a compact set $\mathsf{Y} \subseteq \mathbb{R}^{d_{Y}}$, where $T \in \mathbb{N}$. The observations are assumed to be collected at unit time intervals for notational convenience.
Conditional on the position $X_t$ at time $t \in \mathbb{N}$, the random variable $Y_t$ is assumed to be independent of all other random variables and has a bounded, positive probability density function $g_{\theta}(x_t, y_t)$.
Let $P_{\mu_{t-1,\theta}, t, \theta}(x_{t-1}, dx_t)$ denote the conditional law of $X_t$ given $\mathscr{F}_{t-1}$, which represents the transition kernel over unit time. Let $\nu$ be a probability density on the parameter space $\Theta$ with respect to the Lebesgue measure. Let $d\theta$ denote the Lebesgue measure on $(\Theta, \mathscr{B}(\Theta))$, where $\mathscr{B}(\Theta)$ is the Borel $\sigma$-field on $\Theta$. Our objective is to sample from the posterior distribution
\begin{equation}\label{eq:posterior}
\pi\left(d(\theta, x_{1:T}) \mid y_{1:T}\right) = \frac{L(\theta, x_{1:T}; y_{1:T}) \nu(\theta) d\theta}{\int_{\Theta \times \mathbb{R}^{Td}} L(\theta, x_{1:T}; y_{1:T}) \nu(\theta) d\theta},
\end{equation}
with the likelihood function given by
$$
L(\theta, x_{1:T}; y_{1:T}) = \prod_{k=1}^T g_{\theta}(x_k, y_k) P_{\mu_{k-1,\theta}, k, \theta}(x_{k-1}, dx_k).
$$
We assume throughout that the normalizing constant (denominator) in \eqref{eq:posterior} is finite.

\subsection{Time Discretization}
\label{sec:Time_Discretization}
Since the continuous-time formulation in \eqref{eq:sde} is not directly amenable to computational implementation, we consider a standard time discretization scheme. The resulting discrete-time approximation will require further numerical approximation methods, which we address in Section \ref{sec:method}.


Let $l \in \mathbb{N}_0 = \mathbb{N}_+\cup \{0\}$ be given and define the time step size as $\Delta_l = 2^{-l}$. We employ the first-order Euler-Maruyama scheme to discretize the SDE \eqref{eq:sde}. For $k \in \{0, 1, \ldots, \Delta_l^{-1}T - 1\}$ and initial condition $\widetilde{X}_0 = x_0$, the discretized process satisfies:
\begin{equation}\label{eq:sde_disc}
\widetilde{X}_{(k+1)\Delta_l} = \widetilde{X}_{k\Delta_l} + a_{\theta}\left(\widetilde{X}_{k\Delta_l}, \overline{\zeta}_{1,\theta}(\widetilde{X}_{k\Delta_l}, \mu_{k\Delta_l,\theta}^l)\right)\Delta_l + b_{\theta}\left(\widetilde{X}_{k\Delta_l}, \overline{\zeta}_{2,\theta}(\widetilde{X}_{k\Delta_l}, \mu_{k\Delta_l,\theta}^l)\right)\Delta W_{k},
\end{equation}
where $\Delta W_k = W_{(k+1)\Delta_l} - W_{k\Delta_l}$ denotes the Brownian increment over the time interval $[k\Delta_l, (k+1)\Delta_l]$, and $\mu_{k\Delta_l,\theta}^l$ represents the law of the time-discretized process $\widetilde{X}_{k\Delta_l}$ at time $k\Delta_l$.
 Let $P_{\mu_{t-1,\theta}^l, t, \theta}^l(x_{t-1}, dx_t)$ denote the conditional law of $\widetilde{X}_t$ from the discretized process \eqref{eq:sde_disc}, representing the time-discretized transition kernel over unit time. It corresponds to the discrete-time analog of the continuous-time transition kernel $P_{\mu_{t-1,\theta}, t, \theta}(x_{t-1}, dx_t)$ introduced previously.

Our computational objective becomes sampling from the discretized posterior distribution
\begin{equation}\label{eq:post_disc}
\pi^l\left(d(\theta, x_{1:T}) \mid y_{1:T}\right) = \frac{L^l(\theta, x_{1:T}; y_{1:T}) \nu(\theta) d\theta}{\int_{\Theta \times \mathbb{R}^{Td}} L^l(\theta, x_{1:T}; y_{1:T}) \nu(\theta) d\theta},
\end{equation}
where the discretized likelihood function is given by
$$
L^l(\theta, x_{1:T}; y_{1:T}) = \prod_{k=1}^T g_{\theta}(x_k, y_k) P_{\mu_{k-1,\theta}^l, k, \theta}^l(x_{k-1}, dx_k).
$$
The discretized formulation \eqref{eq:post_disc} approximates the continuous-time posterior \eqref{eq:posterior}, with the approximation error decreasing as the discretization parameter $l$ increases (equivalently, as $\Delta_l \to 0$).
However, even with the Euler-Maruyama time discretization, the evaluation of the transition kernels $P_{\mu_{k-1,\theta}^l, k, \theta}^l(x_{k-1}, dx_k)$ remains computationally intractable due to the mean-field dependence. This necessitates further simulation-based approximation methods, which we develop in the next section.

\section{Methodology}
\label{sec:method}
In this section, we develop the computational methodology for approximate inference under the discretized model introduced earlier. We prepose two PMCMC approaches for posterior sampling: the single-level PMCMC algorithm in Section~\ref{sec:smooth_approx} and its multilevel extension in Section~\ref{sec:smooth_pair_approx}.

\subsection{Single-level PMCMC Algorithm}\label{sec:smooth_approx}

We approximate the laws $\mu_{t,\theta}^l$ for each $t\in\{1,\ldots,T\}$ using the particle-based method from \cite{basic_method}, with parameter $\theta\in\Theta$ fixed. The law approximation procedure corresponds to Step 2(b) of Algorithm \ref{alg:PMCMC_mckean_vlasov}. For each time interval $[t-1,t]$, we initialize $N$ particles $\{X_{t-1}^i\}_{i=1}^N$ drawn from the empirical distribution at the previous time step, or from the initial condition when $t=1$. The particles then evolve on the discrete grid $\{(t-1) + j\Delta_l\}_{j=0}^{\Delta_l^{-1}}$, with interaction terms approximated through empirical averaging. This procedure yields empirical measures $\{\mu_{t-1+j\Delta_l,\theta}^{l,N}\}_{j=1}^{\Delta_l^{-1}}$ that approximate the true McKean-Vlasov laws for subsequent use in particle filtering.

We construct the time-discretized posterior in \eqref{eq:post_disc} based on the particle-based law estimates produced by Algorithm~\ref{alg:PMCMC_mckean_vlasov}. This posterior serves as the target distribution for MCMC sampling. For each $t\in\{1,\dots,T\}$, the transition kernel over the unit time interval is
$$
P_{\mu_{t-1,\theta}^l,t,\theta}^l(x_{t-1},dx_t) = \int_{\mathbb{R}^{d(\Delta_l^{-1}-1)}}
\prod_{k=1}^{\Delta_{l}^{-1}} Q_{\mu_{t-1+(k-1)\Delta_l,\theta}^l,\theta}^l(x_{t-1+(k-1)\Delta_l},dx_{t-1+k\Delta_l}),
$$
where $Q_{\mu_{t-1+(k-1)\Delta_l,\theta}^l,\theta}^l(x_{t-1+(k-1)\Delta_l},dx_{t-1+k\Delta_l})$ denote the Gaussian Markov kernel on $(\mathbb{R}^d,\mathscr{B}(\mathbb{R}^d))$ in $\Delta_l$ time step of the discretized SDE \eqref{eq:sde_disc} with law $\mu_{t-1+(k-1)\Delta_l,\theta}^l\in\mathcal{P}(\mathbb{R}^d)$.

Let $u_t = (x_{t-1+\Delta_l}, \ldots, x_t)$ represent the intermediate states within the time interval $[t-1, t]$ for $t \in \{1, \ldots, T\}$. At discretization level $l$, we introduce the following notation: we define $E_l := (\mathbb{R}^d)^{\Delta_l^{-1}}$ as the state space of discretized paths over unit time intervals for the intermediate states, so that $u_t \in E_l$; for $N$ particles, we define $E_l^N := (E_l)^N$ as the state space for $N$ particle trajectories; for the full time horizon, we define $E_l^T := (E_l)^T$. We then define the extended transition kernel
$$
\overline{P}_{\mu_{t-1,\theta}^l,t,\theta}^l(x_{t-1},du_t) := \prod_{k=1}^{\Delta_{l}^{-1}} Q_{\mu_{t-1+(k-1)\Delta_l,\theta}^l,\theta}^l(x_{t-1+(k-1)\Delta_l},dx_{t-1+k\Delta_l}).
$$
The target posterior distribution on $(\Theta\times\mathsf{E}_l^T,\mathscr{B}(\Theta\times\mathsf{E}_l^T))$ is
\begin{align}
\label{eqn:overline_pi}
\overline{\pi}^l\left(d(\theta,u_1,\dots,u_T)\right) := \frac{\left\{\prod_{k=1}^T g_{\theta}(x_k,y_k)
\overline{P}_{\mu_{k-1,\theta}^l,k,\theta}^l(x_{k-1},du_k)\right\}\nu(d\theta)
}{\int_{\Theta\times\mathsf{E}_l^T}
\left\{\prod_{k=1}^T g_{\theta}(x_k,y_k)
\overline{P}_{\mu_{k-1,\theta}^l,k,\theta}^l(x_{k-1},du_k)\right\}\nu(d\theta)}.
\end{align}
This measure admits \eqref{eq:post_disc} as a marginal. However, direct sampling is intractable since the exact laws $\mu_{k,\theta}^l$ are unknown and must be approximated.

Using the particle-based law approximation procedure from Step 2(b), we replace the true laws with empirical approximations $\mu_{t,\theta}^{l,N}$ based on $N$ particles. Let $\overline{u}_t = u_t^{1:N} \in \mathsf{E}_l^N$
denote all simulated intermediate states $\{(X_{t-1+\Delta_l}^1,\ldots, X_{t-1+\Delta_l}^N), \ldots, (X_t^1,\ldots,X_t^N)\}$ generated by the law approximation procedure at time $t$. The joint distribution of the complete trajectory sequence $(\overline{u}_1,\ldots,\overline{u}_T)$ is denoted by $\overline{\mathbb{P}}_{\theta}^N$, with corresponding expectation operator $\overline{\mathbb{E}}_{\theta}^N$.
The particle-approximated transition kernel is
$$
\overline{P}_{\mu_{t-1,\theta}^{l,N},t,\theta}^l(x_{t-1},du_t) =  \prod_{k=1}^{\Delta_{l}^{-1}} Q_{\mu_{t-1+(k-1)\Delta_l,\theta}^{l,N},\theta}^l(x_{t-1+(k-1)\Delta_l},dx_{t-1+k\Delta_l}).
$$
Our MCMC method targets the approximated posterior
\begin{align} 
\label{eqn:bar_pi}
&\overline{\pi}^{l,N}\left(d(\theta,u_1,\dots,u_T)\right)\\
&= \frac{
\left\{\prod_{k=1}^T g_{\theta}(x_k,y_k)\right\}
\overline{\mathbb{E}}_{\theta}^N\left[\prod_{k=1}^T \overline{P}_{\mu_{k-1,\theta}^{l,N},k,\theta}^l(x_{k-1},du_k)
\right]\nu(\theta)d\theta
}{
\int_{\Theta\times \mathsf{E}_l^T}
\left\{\prod_{k=1}^T g_{\theta}(x_k,y_k)\right\}
\overline{\mathbb{E}}_{\theta}^N\left[
\prod_{k=1}^T\overline{P}_{\mu_{k-1,\theta}^{l,N},k,\theta}^l(x_{k-1},du_k)
\right]\nu(\theta)d\theta
}.\nonumber
\end{align}

\subsubsection{Single-level PMCMC Sampling Algorithm}
\label{subsec:Single-level PMCMC_D}
We employ a single-level PMCMC approach to sample from $\overline{\pi}^{l,N}$. Algorithm \ref{alg:PMCMC_mckean_vlasov} integrates four essential components executed sequentially within each MCMC iteration. The complete algorithmic details are provided in Section \ref{app:single_level_details} of the Supplementary Material.

\begin{enumerate}
\item \textbf{Parameter Proposal} (Step 2(a)): Generate a candidate parameter $\theta'$ from a proposal kernel $r(\theta^{(k-1)}, \theta')$.

\item \textbf{Transition Law Approximation} (Step 2(b)): For the proposed parameter $\theta'$, construct particle-based approximations $\{\mu_{t-1+j\Delta_l,\theta'}^{l,N}\}_{j=1}^{\Delta_l^{-1}}$ of the McKean-Vlasov laws by evolving $N$ particles through discretized dynamics with empirical interaction terms.

\item \textbf{Particle Filter Likelihood Estimation} (Step 2(c)): Conditional on the approximated laws, run a particle filter with $M$ particles to estimate the marginal likelihood $\hat{p}_{\theta'}^{l,M,N}(y_{1:T})$ and generate trajectory samples.

\item \textbf{Metropolis-Hastings Step} (Step 2(d)): Accept or reject the proposed parameter using the acceptance probability
\begin{equation}
\alpha = \min\left\{1,\; \frac{\hat{p}_{\theta'}^{l,M,N}(y_{1:T})\nu(\theta')r(\theta',\theta^{(k-1)})}{\hat{p}_{\theta^{(k-1)}}^{l,M,N}(y_{1:T})\nu(\theta^{(k-1)})r(\theta^{(k-1)},\theta')}\right\}.
\label{eq:alpha_acceptance}
\end{equation}
\end{enumerate}

This acceptance-rejection procedure ensures that the Markov chain $\{\theta^{(k)},\xi_{1:T}^{(k)}\}_{k=0}^K$ has $\overline{\pi}^{l,N}$ (given in \eqref{eqn:bar_pi}) as its stationary distribution. The algorithm is an adaptation of the original PMCMC framework proposed by \cite{andrieu}. Under mild regularity conditions established therein, the resulting Markov chain produces samples from the approximate posterior distribution $\overline{\pi}^{l,N}$.

\textbf{Computational Complexity:} Each MCMC iteration has complexity $\mathcal{O}(\Delta_l^{-1}N^2T) + \mathcal{O}(\Delta_l^{-1}NMT)$. The first term arises from the transition law approximation, where particle interactions are required to compute empirical averages at each discretization step. The second term corresponds to the particle filter operations that involve propagating $M$ filtering particles through transition kernels derived from the $N$ law approximation particles.

\textbf{Implementation Details:} Initialization proceeds by sampling $\theta^{(0)}$ from the prior $\nu(\theta)$ and executing the particle filter to obtain an initial likelihood estimate. The Metropolis-Hastings step is then applied iteratively for $k=1,\ldots,K$ to generate the sample chains $\{\theta^{(k)}, \xi_{1:T}^{(k)}\}_{k=0}^K$.

For any measurable function $\varphi: \Theta \times \mathbb{R}^{Td} \rightarrow \mathbb{R}$, we approximate the expectation
$$
\int_{\Theta\times\mathsf{E}_l^T}\varphi(\theta,x_{1:T}) \overline{\pi}^{l,N}\left(d(\theta,u_1,\dots,u_T)\right)
$$
by the empirical average
\begin{equation}\label{eq:mcmc_est}
\overline{\pi}^{l,N,K}(\varphi) := \frac{1}{K+1}\sum_{k=0}^K \varphi(\theta^{(k)},x_{1:T}^{(k)}).
\end{equation}
This estimator provides consistent approximations to expectations under $\overline{\pi}^{l,N}$ as $K \to \infty$, enabling Monte Carlo inference for the MVSDE parameters.

\subsection{Multilevel PMCMC algorithm}\label{sec:smooth_pair_approx}

\subsubsection{Bi-Level Telescoping Difference}

To implement the MLMC estimation, we need to approximate the bi-level telescoping difference
\begin{equation}\label{eq:diff}
\int_{\Theta\times\mathsf{E}_l^T}\varphi(\theta,x_{1:T})
\overline{\pi}^{l,N_l}\left(d(\theta,u_1,\dots,u_T)\right)
-
\int_{\Theta\times\mathsf{E}_{l-1}^T}\varphi(\theta,x_{1:T})
\overline{\pi}^{l-1,N_{l}}\left(d(\theta,u_1,\dots,u_T)\right),
\end{equation}
where $l\in\mathbb{N}$ is fixed, $\varphi:\Theta\times\mathbb{R}^{Td}\rightarrow\mathbb{R}$ is a test function, and $N_{l}\in\mathbb{N}$ denotes the number of particles used to approximate the law. Effective MLMC estimation necessitates sampling from a coupled distribution of $\overline{\pi}^{l,N_l}$ and $\overline{\pi}^{l-1,N_{l}}$. Here, $\overline{\pi}^{l,N_l}$ is defined analogously to the single-level posterior in \eqref{eqn:bar_pi}, employing $N_l$ particles to approximate the law at discretization level $l$. Similarly, $\overline{\pi}^{l-1,N_{l}}$ corresponds to the posterior at discretization level $l-1$, but crucially utilizes the same particle count $N_l$ to approximate the law. Our bi-level MCMC algorithm, detailed in Step 3 of Algorithm \ref{MLPMCMC_mckean_vlasov_part1}, constructs a joint probability measure on $(\Theta\times\mathsf{E}_{l}^T\times\mathsf{E}_{l-1}^T,\mathscr{B}(\Theta\times\mathsf{E}_{l}^T\times\mathsf{E}_{l-1}^T))$ 
that preserves the correct marginal distributions while introducing correlation across discretization levels. This construction leverages change-of-measure techniques to enable MCMC approximation of \eqref{eq:diff}, and extend the MLMC framework to McKean–Vlasov systems with law-dependent diffusion coefficients.

To implement the bi-level MCMC component, we simultaneously approximate the measure sequences 
$\{\mu_{t-1+\Delta_l,\theta}^l,\ldots,\mu_{t,\theta}^l\}$ and $\{\mu_{t-1+\Delta_{l-1},\theta}^{l-1},\ldots,\mu_{t,\theta}^{l-1}\}$ for each $t \in \{1,\ldots,T\}$. To avoid notational confusion, we denote the latter as $\{\widetilde{\mu}_{t-1+\Delta_{l-1},\theta}^{l-1},\ldots,\widetilde{\mu}_{t,\theta}^{l-1}\}$.
We accomplish this through the coupled law approximation procedure detailed in Step 3(b)(ii) of Algorithm~\ref{MLPMCMC_mckean_vlasov_part1}. This scheme generalizes the standard law approximation from Step 2(b) of Algorithm~\ref{alg:PMCMC_mckean_vlasov} by introducing statistical dependence between fine-level ($l$) and coarse-level ($l-1$) discretizations. The coupling design is inspired by MLMC theory, where correlating approximations across consecutive discretization levels ensures that the variance of the telescoping difference estimator in~\eqref{eq:diff} decreases appropriately as $l$ increases, thereby maintaining the computational efficiency characteristic of effective multilevel methods.

\subsubsection{Coupled Posterior Distribution Framework}
\label{subsec: cou_pos_dis}
We present our bi-level MCMC methodology to approximate \eqref{eq:diff}. We denote by $\check{\mathbb{P}}_{\theta}^{N_l}$ the joint distribution of the coupled trajectory sequences 
\[
(\overline{u}_1^l,\ldots,\overline{u}_T^l)
\quad\text{and}\quad
(\widetilde{\overline{u}}_1^{l-1},\ldots,\widetilde{\overline{u}}_T^{l-1})
\]
over $[0,T]$, with corresponding expectation operator $\check{\mathbb{E}}_{\theta}^{N_l}$. Here, $\overline{u}^l_t = u_t^{1:N_l} \in \mathsf{E}_l^{N_l}$ represents the fine-level trajectory ensemble 
$$\overline{u}^l_t=\Big\{(X_{t-1+\Delta_l}^1,\ldots, X_{t-1+\Delta_l}^{N_l}), \ldots, (X_t^1,\ldots,X_t^{N_l})\Big\},$$
and $\widetilde{\overline{u}}^{l-1}_t = \widetilde{u}_t^{1:N_{l}} \in \mathsf{E}_{l-1}^{N_{l}}$ denotes the coarse-level trajectory ensemble $$\widetilde{\overline{u}}^{l-1}_t=\Big\{(\widetilde{X}_{t-1+\Delta_{l-1}}^1,\ldots, \widetilde{X}_{t-1+\Delta_{l-1}}^{N_l}), \ldots, (\widetilde{X}_t^1,\ldots,\widetilde{X}_t^{N_l})\Big\},$$ generated through the coupled law approximation procedure (Algorithm 2, Step 3(b)(ii)). The superscript $N_l$ explicitly indicates that this distribution arises from using $N_l$ particles to construct empirical approximations $\{\mu^{l,N_l}_{t,\theta}\}$ and $\{\tilde{\mu}^{l-1,N_l}_{t,\theta}\}$ of the McKean-Vlasov laws at discretization levels $l$ and $l-1$. 

For observation weighting, we define the coupled conditional density
\begin{equation}\label{eq:hk_ch}
	H_{k,\theta}(x,x';y_k) = \tfrac{1}{2}\{g_{\theta}(x',y_k) + g_{\theta}(x,y_k)\},
\end{equation}
for $(x,x',\theta)\in\mathbb{R}^{2d}\times\Theta$ and $k\in\{1,\ldots,T\}$.  
This symmetric averaging introduces correlation between discretization levels. Algorithm~\ref{MLPMCMC_mckean_vlasov_part1} produces coupled kernels for simultaneous simulation of $
 \overline{P}_{\mu_{t-1,\theta}^{l,N_l},t,\theta}^l(x_{t-1},du_t)$ and 
$ \overline{P}_{\mu_{t-1,\theta}^{l-1,N_{l}},t,\theta}^{l-1}(x_{t-1},du_t)
$,
where the underlying probability measures are generated through the pre-computed empirical approximations in Step 3(b)(ii) of Algorithm~\ref{MLPMCMC_mckean_vlasov_part1}. We denote by $\check{P}^{l}_{\theta}$ the coupled transition kernel that generates $(u_t^l, \widetilde{u}_{t}^{l-1}) \in \mathsf{E}_l \times \mathsf{E}_{l-1}$ at time $t$, conditional on the previous state $(u_{t-1}^l, \widetilde{u}_{t-1}^{l-1})$ at time $t-1$, with initial condition $(u_{0}^l, \widetilde{u}_{0}^{l-1}) = (x_0, x_0)$. This kernel $\check{P}^{l}_{\theta}$ is employed in the Delta Particle Filter (Algorithm 2, Step 3(b)(iii)). At this stage, the coupled empirical measures 
$\{\mu^{l,N_l}_{s,\theta}\}$ and $\{\tilde{\mu}^{l-1,N_l}_{s,\theta}\}$ are fixed and serve as deterministic input parameters.

The joint probability measure, serving as the target distribution for our bi-level MCMC sampler, is defined on the product space $\left(\Theta\times(\mathsf{E}_l\times\mathsf{E}_{l-1})^T,\mathscr{B}(\Theta\times(\mathsf{E}_l\times\mathsf{E}_{l-1})^T)\right)$ as
\begin{align}\label{eq:main_tar}
&\check{\pi}^{l,N_{l}}\!\left(d(\theta,u_{1:T}^l,\widetilde{u}_{1:T}^{l-1})\mid y_{1:T}\right)\\
&= 
\dfrac{
\splitdfrac{
\left\{\prod_{k=1}^T H_{k,\theta}(x_k^l,\widetilde{x}_k^{l-1};y_k)\right\}
\check{\mathbb{E}}_{\theta}^{N_{l}}\!\left[
\prod_{k=1}^T 
\check{P}_{\theta}^{N_l}\!\left(
(x_{k-1}^l,\widetilde{x}_{k-1}^{l-1}),
d(u_k^l,\widetilde{u}_{k}^{l-1})
\right)\right]}
{\times\nu(\theta)d\theta}
}{
\splitdfrac{
\int_{\Theta\times(\mathsf{E}_l\times\mathsf{E}_{l-1})^T}
\!\left\{\prod_{k=1}^T H_{k,\theta}(x_k^l,\widetilde{x}_k^{l-1};y_k)\right\}
\check{\mathbb{E}}_{\theta}^{N_{l}}\!\left[
\prod_{k=1}^T 
\check{P}_{\theta}^{N_l}\!\left(
(x_{k-1}^l,\widetilde{x}_{k-1}^{l-1}),
d(u_k^l,\widetilde{u}_{k}^{l-1})
\right)\right]}
{\times\nu(\theta)d\theta}
}.\nonumber
\end{align}

The coupled measure $\check{\pi}^{l,N_{l}}$ encodes the statistical dependence between trajectories at discretization levels $l$ and $l-1$, which is essential for effective variance reduction in multilevel methods.
To enable telescoping decomposition, we define the auxiliary weight
\[
\check{H}_{k,\theta}(x,x';y_k)
= \frac{g_{\theta}(x,y_k)}{H_{k,\theta}(x,x';y_k)},
\qquad (x,x',\theta)\in\mathbb{R}^{2d}\times\Theta,\ k\in\{1,\ldots,T\}.
\]
This ratio connects the fine-level observation density $g_{\theta}(x,y_k)$ with the averaged coupled density $H_{k,\theta}$ from~\eqref{eq:hk_ch}. 

With these definitions established, the bi-level telescoping difference in~\eqref{eq:diff} can now be expressed in terms of expectations with respect to the coupled measure $\check{\pi}^{l,N_{l}}$. The key technical challenge is that $\check{\pi}^{l,N_l}$ uses the averaged observation density $H_{k,\theta}$ to introduce correlation between levels for variance reduction, whereas our target quantities $\bar{\pi}^{l,N_l}(\varphi)$ and $\bar{\pi}^{l-1,N_l}(\varphi)$ involve the original observation densities $g_{\theta}$. The auxiliary weights $\check{H}_{k,\theta}$ enable the necessary change of measure between these distributions. Specifically, the importance weight $\prod_{k=1}^T \check{H}_{k,\theta}(x_k^l,\widetilde{x}_{k}^{l-1};y_k)$ has the critical property that its denominator $\prod_{k=1}^T H_{k,\theta}(x_k^l,\widetilde{x}_{k}^{l-1};y_k)$ exactly matches the coupled observation density term in the coupled measure $\check{\pi}^{l,N_l}$, which enables a cancellation: 
\begin{align*}
&\dfrac{
\splitdfrac{
\int_{\Theta\times(\mathsf{E}_l\times\mathsf{E}_{l-1})^T}\!\varphi(\theta,x_{1:T}^l)
\!\left\{\frac{\prod_{k=1}^T g_{\theta}(x_k^l,y_k)}{\prod_{k=1}^T H_{k,\theta}(x_k^l,\widetilde{x}_{k}^{l-1};y_k)}\right\}
\!\left\{\prod_{k=1}^T H_{k,\theta}(x_k^l,\widetilde{x}_k^{l-1};y_k)\right\}}
{
\times\check{\mathbb{E}}_{\theta}^{N_{l}}\!\left[
\prod_{k=1}^T \check{P}_{\theta}^{l}\!\left((x_{k-1}^l,\widetilde{x}_{k-1}^{l-1}),d(u_k^l,\widetilde{u}_{k}^{l-1})\right)
\right]\nu(\theta)d\theta}
}
{
\splitdfrac{
\int_{\Theta\times(\mathsf{E}_l\times\mathsf{E}_{l-1})^T}
\!\left\{\frac{\prod_{k=1}^T g_{\theta}(x_k^l,y_k)}{\prod_{k=1}^T H_{k,\theta}(x_k^l,\widetilde{x}_{k}^{l-1};y_k)}\right\}
\!\left\{\prod_{k=1}^T H_{k,\theta}(x_k^l,\widetilde{x}_k^{l-1};y_k)\right\}}
{
\times\check{\mathbb{E}}_{\theta}^{N_{l}}\!\left[
\prod_{k=1}^T \check{P}_{\theta}^{l}\!\left((x_{k-1}^l,\widetilde{x}_{k-1}^{l-1}),d(u_k^l,\widetilde{u}_{k}^{l-1})\right)
\right]\nu(\theta)d\theta}
}.
\end{align*}
The cancellation restores the fine-level posterior $\bar{\pi}^{l,N_l}$ with the original observation densities $\prod_{k=1}^T g_{\theta}(x_k^l,y_k)$. 

Thus, we sample coupled trajectories $((x_1^l, \ldots, x_T^l), (\widetilde{x}_1^{l-1}, \ldots, \widetilde{x}_T^{l-1}))$ from the coupled measure $\check{\pi}^{l,N_l}$, which employs the averaged observation density $H_{k,\theta}(x_k^l, \widetilde{x}_k^{l-1}; y_k)$ to induce strong correlation between discretization levels. These trajectories are then reweighted using importance weights $\prod_{k=1}^T \check{H}_{k,\theta}(x_k^l, \widetilde{x}_k^{l-1}; y_k)$ and $\prod_{k=1}^T \check{H}_{k,\theta}(\widetilde{x}_k^{l-1}, x_k^l; y_k)$ to compute expectations under the correct marginal posteriors $\bar{\pi}^{l,N_l}$ and $\bar{\pi}^{l-1,N_l}$. This change-of-measure technique is key to variance reduction in the multilevel framework: coupled sampling ensures high correlation between fine- and coarse-level trajectories, reducing variance, while importance weights guarantee expectations align with the target distributions. 

The telescoping difference could be rewritten as
$$
\frac{\int_{\Theta\times(\mathsf{E}_l\times\mathsf{E}_{l-1})^T}\varphi(\theta,x_{1:T}^l)
\left\{\prod_{k=1}^T \check{H}_{k,\theta}(x_k^l,\widetilde{x}_{k}^{l-1};y_k)\right\}
\check{\pi}^{l,N_{l}}\left(d(\theta,u_{1:T}^l,\widetilde{u}_{1:T}^{l-1})\right)
}
{\int_{\Theta\times(\mathsf{E}_l\times\mathsf{E}_{l-1})^T}
\left\{\prod_{k=1}^T \check{H}_{k,\theta}(x_k^l,\widetilde{x}_{k}^{l-1};y_k)\right\}
\check{\pi}^{l,N_{l}}\left(d(\theta,u_{1:T}^l,\widetilde{u}_{1:T}^{l-1})\right)} 
$$
\begin{equation}\label{eq:main_eq}
-\frac{\int_{\Theta\times(\mathsf{E}_l\times\mathsf{E}_{l-1})^T}\varphi(\theta,\widetilde{x}_{1:T}^{l-1})
\left\{\prod_{k=1}^T \check{H}_{k,\theta}(\widetilde{x}_{k}^{l-1},x_k^l;y_k)\right\}
\check{\pi}^{l,N_{l}}\left(d(\theta,u_{1:T}^l,\widetilde{u}_{1:T}^{l-1})\right)
}
{\int_{\Theta\times(\mathsf{E}_l\times\mathsf{E}_{l-1})^T}
\left\{\prod_{k=1}^T \check{H}_{k,\theta}(\widetilde{x}_{k}^{l-1},x_k^l;y_k)\right\}
\check{\pi}^{l,N_{l}}\left(d(\theta,u_{1:T}^l,\widetilde{u}_{1:T}^{l-1})\right)}.
\end{equation}

With the help of \eqref{eq:main_eq}, we are going to develop an efficient MCMC algorithm for sampling from the coupled target measure \eqref{eq:main_tar}. The key challenge lies in designing an effective coupling between the discretized posterior distributions $\overline{\pi}^{l,N_l}$ and $\overline{\pi}^{l-1,N_{l}}$ using MCMC techniques. No standard MCMC methodology exists for this purpose beyond trivial independent coupling, which fails to provide necessary variance reduction. To tackle it, we construct the joint probability measure \eqref{eq:main_tar} that induces statistical dependence between the coupled trajectories $(u_{1:T}^l,\widetilde{u}_{1:T}^{l-1})$. 

\subsubsection{Multilevel PMCMC Sampling Algorithm}
\label{subsec:Multilevel-level PMCMC_D}

Our multilevel PMCMC methodology combines single-level and bi-level components to estimate the MLMC approximation. The framework integrates single-level PMCMC for base level computation with bi-level MCMC for telescoping difference estimation across multiple discretization levels. The complete algorithmic details are provided in Section \ref{app:multilevel_details} of the Supplementary Material. We now provide an overview of Algorithm \ref{MLPMCMC_mckean_vlasov_part1}:

\textbf{1. Base Level Computation} (Step 2): Execute Algorithm~\ref{alg:PMCMC_mckean_vlasov} at base discretization level $l_{\star} \in \mathbb{N}$ for $K_{l_{\star}}$ iterations with $N_{l_{\star}}$ particles in law approximation to obtain base-level samples $\{\theta^{(k)}_{l_{\star}}, x^{(k)}_{l_{\star},1:T}\}_{k=0}^{K_{l_{\star}}}$ and the corresponding likelihood estimation.

\textbf{2. Bi-level Telescoping Difference Approximation} (Step 3): For each level $l \in \{l_{\star}+1,\ldots,L\}$, we independently run the bi-level MCMC to compute the corresponding telescoping difference.

\textbf{2.1. Parameter Proposal} (Step 3(b)(i)): For each level $l$, generate candidate parameter $\theta'$ from proposal kernel $r(\theta^{(k-1)}, \theta')$.

\textbf{2.2. Coupled Law Approximation} (Step 3(b)(ii)): For the proposed parameter $\theta'$, construct synchronized particle-based approximations $\{\mu_{t-1+j\Delta_l,\theta'}^{l,N_l}\}_{j=1}^{\Delta_l^{-1}}$ and $\{\widetilde{\mu}_{t-1+j\Delta_{l-1},\theta'}^{l-1,N_l}\}_{j=1}^{\Delta_{l-1}^{-1}}$ of the McKean--Vlasov laws at consecutive discretization levels by evolving $N_l$ coupled particle pairs with correlated Brownian increments $W$ and $\widetilde{W}$, where $\Delta \widetilde{W}_{t-1+j\Delta_{l-1}}^i = \Delta W_{t-1+(2j-1)\Delta_l}^i + \Delta W_{t-1+2j\Delta_l}^i$.

\textbf{2.3. Delta Particle Filter Likelihood Estimation} (Step 3(b)(iii)): For a given parameter $\theta' \in \Theta$, conditional on the sequence of coupled approximated empirical measures $\{\mu^{l,N_l}_{t,\theta'}, \widetilde{\mu}^{l-1,N_l}_{t,\theta'}\}$ generated in Step 3(b)(ii), we employ a Delta Particle Filter with $M$ coupled particle pairs to estimate the joint likelihood $\hat{p}^{l,M,N_l}_{\theta'}(y_{1:T})$ and generate coupled trajectory samples for subsequent use in the Metropolis-Hastings step.

\textbf{2.4. Metropolis-Hastings Step} (Step 3(b)(iv)): Accept or reject the proposed parameter using the acceptance probability
\begin{align}
    \alpha = \min\left\{1,\; \frac{\hat{p}^{l,M,N_l}_{\theta'}(y_{1:T})\nu(\theta')r(\theta',\theta^{(k-1)})}{\hat{p}^{l,M,N_l}_{\theta^{(k-1)}}(y_{1:T})\nu(\theta^{(k-1)})r(\theta^{(k-1)},\theta')}\right\}.
    \label{eq:alpha}
\end{align}

\textbf{2.5. Importance Weight Correction for Telescoping Difference}: Apply importance weighting correction using the auxiliary weight function $\check{H}_{k,\theta}(x,x';y_k) = \frac{g_\theta(x,y_k)}{H_{k,\theta}(x,x';y_k)}$ to obtain samples from the correct posterior difference $\bar{\pi}^{l,N_l}(\varphi) - \bar{\pi}^{l-1,N_l}(\varphi)$, as rigorously justified through equation \eqref{eq:main_eq}.

\textbf{Bi-level MCMC Computational Complexity}: Each bi-level MCMC iteration has complexity $\mathcal{O}(\Delta_l^{-1}T\{MN_l + N_l^2\})$. The $\mathcal{O}(\Delta_l^{-1}N_l^2 T)$ term arises from coupled law approximation, while $\mathcal{O}(\Delta_l^{-1}MN_l T)$ corresponds to Delta Particle Filter operations.

\textbf{Bi-level Difference Estimator}: Using coupled samples $\{\theta^l(k), u^l_{1:T}(k), \widetilde{u}^{l-1}_{1:T}(k)\}_{k=0}^{K_l}$ generated by the bi-level MCMC for each level $l \in \{l_{\star}+1,\ldots,L\}$, we approximate the telescoping difference through
\hspace{1cm}\begin{align}\label{eq:mcmc_bl_est}
&\hspace{-1cm}\overline{\pi}^{l,N_l,K_l}(\varphi) - \overline{\pi}^{l-1,N_{l},K_l}(\varphi) \\:= &
\frac{
\frac{1}{K_l+1}\sum_{k=0}^{K_l}\varphi(\theta^l(k),x_{1:T}^{l}(k))
\left\{\prod_{s=1}^T \check{H}_{s,\theta^l(k)}(x_s^l(k),\widetilde{x}_{s}^{l-1}(k);y_s)\right\}
}{\frac{1}{K_l+1}\sum_{k=0}^{K_l}
\left\{\prod_{s=1}^T \check{H}_{s,\theta^l(k)}(x_s^l(k),\widetilde{x}_{s}^{l-1}(k);y_s)\right\}} \nonumber\\
&- \frac{
\frac{1}{K_l+1}\sum_{k=0}^{K_l}\varphi(\theta^l(k),\widetilde{x}_{1:T}^{l-1}(k))
\left\{\prod_{s=1}^T \check{H}_{s,\theta^l(k)}(\widetilde{x}_{s}^{l-1}(k),x_s^l(k);y_s)\right\}
}{\frac{1}{K_l+1}\sum_{k=0}^{K_l}
\left\{\prod_{s=1}^T \check{H}_{s,\theta^l(k)}(\widetilde{x}_{s}^{l-1}(k),x_s^l(k);y_s)\right\}}. \nonumber
\end{align}

\textbf{Complete Multilevel Framework}: For test function $\varphi:\Theta\times\mathbb{R}^{Td}\rightarrow\mathbb{R}$, define
$$
\pi(\varphi) := \int_{\Theta\times\mathbb{R}^{Td}}\varphi(\theta,x_{1:T}) \pi\left(d(\theta,x_{1:T})\, | \, y_{1:T}\right),
$$
which we assume is finite. The multilevel estimator is
\begin{equation}\label{eq:mlmc_final_estimator}
\widehat{\pi}(\varphi) := \overline{\pi}^{l_{\star},N_{l_{\star}},K_{l_{\star}}}(\varphi) + \sum_{l=l_{\star}+1}^L\left\{
\overline{\pi}^{l,N_l,K_l}(\varphi) - \overline{\pi}^{l-1,N_{l},K_l}(\varphi)
\right\},
\end{equation}
where $\overline{\pi}^{l_{\star},N_{l_{\star}},K_{l_{\star}}}(\varphi)$ is computed via \eqref{eq:mcmc_est} and each difference $\overline{\pi}^{l,N_l,K_l}(\varphi) - \overline{\pi}^{l-1,N_{l},K_l}(\varphi)$ is computed via \eqref{eq:mcmc_bl_est}.
The estimator is subject to six sources of error: base-level MCMC error, multilevel MCMC error, cross-level interaction error, base-level particle approximation error, multilevel particle approximation error, and time discretization bias. Through appropriate selection of parameters $L$, $\{K_{l_{\star}},\ldots,K_L\}$, and $\{N_{l_{\star}},\ldots,N_L\}$, these errors can be systematically controlled, as analyzed in the next section.

\section{Theoretical Results}\label{sec:theory}

In this section, we establish the theoretical results of the proposed methodology. Our analysis relies on three regularity assumptions (A\ref{ass:1})–(A\ref{ass:3}), detailed in Section~\ref{app:assumptions} of the Supplementary Material. The proofs build upon the techniques in Appendix~\ref{app:appendix}, with extensions to accommodate the law-dependent diffusion coefficients. Throughout this section, $C$ denotes a generic finite constant whose value may vary from line to line.

\subsection{Main Results}
\label{sec:main_result}
Based on the regularity assumptions above and auxiliary technical results from Section~\ref{app:appendix} of the Supplementary Material, we present the main theoretical contribution of this work. Throughout this analysis, $\mathbb{E}$ denotes the expectation operator with respect to the randomness defined in (\ref{eq:mlmc_final_estimator}). For notational simplicity, we omit the time horizon $T$ from our expressions.

\begin{theorem}[MSE Bound for Multilevel PMCMC Estimator]\label{thm:main_convergence}
Under Assumptions (A\ref{ass:1})-(A\ref{ass:3}), let $\varphi\in\mathcal{C}_b^2(\Theta\times\mathbb{R}^{dT})\cap\mathcal{B}_b(\Theta\times\mathbb{R}^{dT})$ be a test function, there exists a constant $C<\infty$ such that for any multilevel configuration specified by integers $l_{\star}<L$ and particle/iteration counts $(N_{l_{\star}},\ldots,N_L,K_{l_{\star}},\ldots,K_L)\in\mathbb{N}^{2(L-l_{\star}+1)}$, the mean square error satisfies
\begin{equation}\label{eq:mse_decomposition}
\mathbb{E}\left[\left(\widehat{\pi}(\varphi) - \pi(\varphi)\right)^2\right] \leq C\left(\varepsilon_{\mathrm{MCMC}}^{(l_{\star})} + \varepsilon_{\mathrm{MCMC}}^{(\mathrm{ML})} + \varepsilon_{\mathrm{cross}} + \varepsilon_{\mathrm{part}}^{(l_{\star})} + \varepsilon_{\mathrm{part}}^{(\mathrm{ML})} + \varepsilon_{\mathrm{disc}}\right),
\end{equation}
where the error components are defined as follows:
\begin{itemize}
\item Base-level MCMC error: $\displaystyle \varepsilon_{\mathrm{MCMC}}^{(l_{\star})} := \frac{1}{K_{l_{\star}}+1}$,

\item Multilevel MCMC error: $\displaystyle \varepsilon_{\mathrm{MCMC}}^{(\mathrm{ML})} := \sum_{l=l_{\star}+1}^L\frac{\Delta_l}{K_l+1}$,

\item Cross-level interaction error: $\displaystyle \varepsilon_{\mathrm{cross}} := \sum_{\substack{l,q=l_{\star}+1\\l\neq q}}^L\frac{\sqrt{\Delta_l\Delta_q}}{(K_l+1)(K_q+1)}$,

\item Base-level particle approximation error: $\displaystyle \varepsilon_{\mathrm{part}}^{(l_{\star})} := \frac{1}{N_{l_{\star}}}$,

\item Multilevel particle approximation error: $\displaystyle \varepsilon_{\mathrm{part}}^{(\mathrm{ML})} := \left(\sum_{l=l_{\star}+1}^L\frac{\sqrt{\Delta_l}}{\sqrt{N_l}}\right)^2$,

\item Time discretization bias: $\displaystyle \varepsilon_{\mathrm{disc}} := \Delta_L^2$.
\end{itemize}
\end{theorem}

The proof of Theorem \ref{thm:main_convergence} is provided in Section \ref{app:main_proof} of the Supplementary Material.
The practical implications of Theorem \ref{thm:main_convergence} can be quantified as follows. Let $\epsilon > 0$ denote a target accuracy level. Under the standard assumption that $M = \mathcal{O}(T)$ particles are employed in the particle filter components (Steps 2(c) of Algorithm \ref{alg:PMCMC_mckean_vlasov} and Steps 3(b)(iii) of Algorithm \ref{MLPMCMC_mckean_vlasov_part1}), as is typical in practice \citep{andrieu}, and neglecting constants associated with the observation horizon $T$, the total computational cost of our multilevel estimator \eqref{eq:mlmc_final_estimator} is $\mathcal{O}\left(\sum_{l=l_{\star}}^L I_l \Delta_l^{-1} N_l^2\right)$.
Within the MLMC framework, where the computational burden of the base level $l_{\star}$ becomes negligible, optimal parameter selection yields $I_l = \mathcal{O}(\epsilon^{-2}\Delta_l^{6/7})$, $N_l = \mathcal{O}(\epsilon^{-2}\Delta_l^{1/2})$, and $L = \mathcal{O}(|\log(\epsilon)|)$. This configuration achieves an upper bound on the MSE $\mathbb{E}\left[\left(\widehat{\pi}(\varphi) - \pi(\varphi)\right)^2\right]$ of $\mathcal{O}(\epsilon^2)$ while incurring computational cost $\mathcal{O}(\epsilon^{-6})$.
By contrast, a conventional single-level approach employing only level $L$ via Algorithm \ref{alg:PMCMC_mckean_vlasov} requires computational cost $\mathcal{O}(\epsilon^{-7})$ to attain the same mean square error of $\mathcal{O}(\epsilon^2)$; i.e. we reduce the cost by an order of magnitude. Consequently, our multilevel methodology achieves an order-of-magnitude reduction in computational complexity relative to standard single-level PMCMC approaches.

\section{Numerical Simulations}\label{sec:numerics}
This section introduces a neuronal activity model that utilizes a multidimensional MVSDE with a non-globally Lipschitz structure to capture complex neural dynamics. The model’s three-dimensional formulation is elaborated in Section \ref{sec:3DNeuron}. Numerical simulation results are presented in Section \ref{sec:Simulation}.

\subsection{3D Neuron Model}
\label{sec:3DNeuron}

This study examines the neuronal activity model from \cite{dos2022simulation}, employing a MVSDE with non-globally Lipschitz structure. The model operates in three dimensions and incorporates specifically designed drift and diffusion terms to capture neuronal dynamics. The associated MVSDE adopts a general form, 
\begin{equation}
\label{eq:3d}
dX_t = a( t, x_t, \mu_t) dt +  b( t, x_t, \mu_t) dW_t,
\end{equation}
where the drift term and diffusion coefficient are expressed as follows:
\begin{align*}
a \left( t, x_t, \mu_t \right) 
&:= \left( \begin{array}{c}
x_1 - (x_1)^3 / 3 - x_2 + I - \int_{\mathbb{R}^3} J \left( x_1 - V_{rev} \right) z_3 d \mu (z) \vspace{0.1cm} \\
c \left( x_1 + a - b x_2 \right) \vspace{0.1cm}\\
a_r \frac{ T_{max} (1 - x_3) }{ 1 + \exp ( - \lambda ( x_1 - V_T)) } - a_d x_3
\end{array} \right), \\
b \left( t, x_t, \mu_t \right) 
&:= \left( \begin{array}{ccc}
b_{ext} & 0 & - \int_{\mathbb{R}^3} b_J \left( x_1 - V_{rev} \right) z_3 d \mu (z) \vspace{0.1cm}\\
0 & 0 & 0 \vspace{0.1cm}\\
0 & b_{32} (x) & 0
\end{array} \right),
\end{align*}
with
\begin{equation*}
b_{32} (x) := \mathbbm{1}_{\{x_{3} \in (0,1)\}} \sqrt{ a_r \frac{ T_{max} (1 - x_3) }{ 1 + \exp ( - \lambda ( x_1 - V_T)) } + a_d x_3 } \ \Gamma \exp ( - \Lambda / (1 - (2x_3 - 1)^2)).
\end{equation*}
In the mean-field interaction terms $\int_{\mathbb{R}^3} J(x_1 - V_{\text{rev}}) z_3 \, d\mu(z)$ and $\int_{\mathbb{R}^3} b_J(x_1 - V_{\text{rev}}) z_3 \, d\mu(z)$ appearing in the drift and diffusion coefficients, the variable $z = (z_1, z_2, z_3)$ represents the state of an arbitrary neuron in the population at time $t$, distinct from the focal neuron state $x = (x_1, x_2, x_3)$. Since the integration is with respect to $z$, the functions $J(x_1 - V_{\text{rev}})$ and $b_J(x_1 - V_{\text{rev}})$ can be treated as constants (depending only on the focal neuron state $x_1$), and the expectation is taken with respect to the third component $z_3$ of the population distribution. The two mean-field interaction terms $\int_{\mathbb{R}^3} J (x_1 - V_{rev}) z_3 d\mu(z)$ and $\int_{\mathbb{R}^3} b_J (x_1 - V_{rev}) z_3 d\mu(z)$ represent the population-averaged synaptic input to each neuron, where $z_3$ denotes the synaptic gating variable of other neurons in the network, creating the McKean-Vlasov coupling structure.

For a detailed description of the physical interpretation of state variables and model parameters, see Section~\ref{app:neuron_model_details} of the Supplementary Material.


We implemented numerical simulations for the aforementioned model using discrete observations $\{Y_k\}_{k=1}^{T}$ recorded at unit time intervals. The observation process follows a conditional Gaussian distribution given by
\begin{equation}
Y_k | X_k = x_k \sim \mathcal{N}(x_k, \Sigma), \quad k \in \{1, 2, \ldots, T\},
\end{equation}
where $T=50$ is selected as the final time, $\mathcal{N}(x_k, \Sigma)$ denotes a multivariate normal distribution with mean vector $\mu = x_k$ and diagonal covariance matrix $\Sigma = \text{diag}(\sigma_1^2, \sigma_2^2, \sigma_3^2)$, with $\sigma_1^2, \sigma_2^2, \sigma_3^2$ denoting the variances of the observation noise for the three dimensions respectively.

The initial condition is specified as follows:
\begin{equation*}
X_0 \sim \mathcal{N} \left( \left( \begin{array}{c} V_0 \\ w_0 \\ y_0 \end{array} \right) , \left( \begin{array}{ccc} \sigma_{V_0} & 0 & 0 \\ 0 & \sigma_{w_0} & 0 \\ 0 & 0 & \sigma_{y_0} \end{array} \right) \right).
\end{equation*}
This simulation study employs the true parameter values listed in Table \ref{tab:model_parameters} of the Supplementary Material to generate synthetic data, with the observation noise levels fixed at $(\sigma_1, \sigma_2, \sigma_3) = (0.2, 0.1, 0.02)$. These distinct variance settings are specifically selected to accommodate the varying numerical magnitudes inherent to the membrane potential $V$, the recovery variable $w$, and the gating variable $y$ within the neuron model.

\subsection{Simulation Results}
\label{sec:Simulation}

We treated the parameters $\{I, J, c, \lambda, b_{ext}, \Gamma, \sigma_1, \sigma_2, \sigma_3\}$ as unknowns to be inferred, chosen for their minimal cross-parameter dependencies. Independent Gaussian priors are assigned to the parameters $I$ and $J$. As $\{c, \lambda, b_{ext}, \Gamma, \sigma_1, \sigma_2, \sigma_3\}$ must be positive, we infer their log-transformed counterparts, $\log(c)$, $\log(\lambda)$, $\log(b_{ext})$, $\log(\Gamma)$, $\log(\sigma_1)$, $\log(\sigma_2)$ and $\log(\sigma_3)$, each assigned independent Gaussian prior distributions. The MCMC sampling utilizes Gaussian random walk proposals with a diagonal covariance matrix structure. The sampler is tuned to optimize mixing properties and ensure robust convergence diagnostics. Simulation parameters adhere to the specifications detailed in Section \ref{sec:theory}.

To assess computational efficiency, we conducted a comparative analysis of the logarithmic relationship between computational cost and mean squared error (MSE) for both the PMCMC and MLPMCMC algorithms applied to the 3D Neuron Model. Specifically, we analyzed how the MSE decreases as a function of computational cost by examining the relationship between $\log(\mathrm{COST})$ and $\log(\mathrm{MSE})$. For each algorithm, the convergence rate is estimated numerically as the slope of the linear regression line fitted to the data points $(\log(\mathrm{COST}), \log(\mathrm{MSE}))$ obtained from multiple simulation runs over a range of discretization levels. This slope quantifies the rate at which the MSE decreases as computational cost increases, providing a direct measure of algorithmic efficiency.

The estimated convergence rates presented in Table \ref{tab:rate_3d} of the Supplementary Material reveal convergence rates of approximately $-3$ for the multilevel method, corresponding to a computational complexity of $\mathcal{O}(\epsilon^{-6})$, while the single-level method exhibits rates near $-3.5$, indicating a complexity of $\mathcal{O}(\epsilon^{-7})$. 
The empirically observed rates show good agreement with the theoretical analysis; the slight discrepancies are expected due to inherent uncertainties in estimating the underlying distribution. The superior computational scaling of the multilevel approach demonstrates its efficiency advantage over the conventional single-level methodology.

The convergence behavior of the proposed methods applied to the 3D neuron model is illustrated in the Supplementary Material: Figure~\ref{fig:trace_d3_pmcmc} presents the trace plots for PMCMC, Figure~\ref{fig:Trace_d3_ml} presents the trace plots for MLPMCMC, and Figure~\ref{fig:incre_d3_ml2} displays the running mean of increments for MLPMCMC. With 10000 iterations, both PMCMC and MLPMCMC exhibit satisfactory mixing properties across the parameter space, as evidenced by the trace plots. Furthermore, the running mean of increments for MLPMCMC demonstrates gradual convergence, confirming the ergodicity of the multilevel estimator.




\section{Conclusion}
\label{sec:Conclusion}
In this work, we have developed a general Bayesian framework for the estimation of latent states and parameters in PO-MVSDEs with fully law-dependent drift and diffusion coefficients. Our methodology introduces the first PMCMC and MLPMCMC algorithms capable of handling this general class of mean-field models. A key innovation lies in the multilevel coupling strategy, which synchronizes the evolution of particle-system laws across discretization levels and constructs correlated likelihood estimators through a coupled Delta Particle Filter. This design preserves the variance reduction properties characteristic of MLMC while accommodating the nonlinearities introduced by law-dependent stochasticity.

We established that the proposed MLPMCMC algorithm achieves an optimal mean square error of order 
 $O(\varepsilon^2)$ with computational cost $O(\varepsilon^{-6})$, improving upon the $O(\varepsilon^{-7})$ complexity of single-level PMCMC. These results hold under standard regularity and ellipticity assumptions on the law-dependent diffusion coefficient, demonstrating the robustness of the multilevel methodology even in the presence of distribution-dependent volatility. Numerical experiments on a three-dimensional neuron model with population-level noise further validate our theoretical findings and highlight the practical relevance of our approach for modeling systems where realistic volatility structures are essential.

Overall, this work significantly broadens the applicability of multilevel PMCMC methods to interaction-driven stochastic systems and provides a principled path forward for scalable, accurate inference in models where both drift and diffusion exhibit mean-field dependence. Future research may extend these techniques to high-dimensional interacting systems, jump–diffusion settings, or adaptive multilevel schemes that automatically tune discretization and particle allocations (see, e.g., \cite{ning2023iterated,ionides2024iterated}).

\section*{Supplementary Materials}
The supplementary material provides detailed descriptions of the single-level and multilevel PMCMC algorithms, complete pseudocode, detailed description of the 3D Neuron Model, additional tables and figures, regularity assumptions, and full proofs of the theoretical results.

\section*{Disclosure Statement}
The authors report there are no other competing interests to declare.

\section*{Data Availability Statement}
The code and data to reproduce our numerical results are publicly available at \url{https://github.com/aminwu-111/MV_SDE_FULL}.

\bibliography{bibliography.bib}
\appendix

\clearpage
\setcounter{page}{1}
\renewcommand{\thepage}{S\arabic{page}}

\renewcommand{\thesection}{S}
\renewcommand{\theequation}{S.\arabic{equation}}
\renewcommand{\thelemma}{S.\arabic{lemma}}
\renewcommand{\theassumption}{S.\arabic{assumption}}
\renewcommand{\thetheorem}{S.\arabic{theorem}}
\renewcommand{\thefigure}{F.\arabic{figure}}
\setcounter{figure}{0}
\renewcommand{\thetable}{T.\arabic{table}}
\setcounter{table}{0}
\section{Supplementary Materials}

\renewcommand{\theequation}{S.\arabic{equation}}
\setcounter{equation}{0}

\startcontents[supp]
\printcontents[supp]{l}{1}{\section*{Supplementary Table of Contents}}

\subsection{Single-level PMCMC Algorithm Details}
\label{app:single_level_details}

This section provides the complete algorithmic details for the single-level PMCMC sampling algorithm described in Section \ref{subsec:Single-level PMCMC_D} of the main text.

1. \textbf{Parameter Proposal} (Step 2(a)): Generate a candidate parameter $\theta'$ from a proposal kernel $r(\theta^{(k-1)}, \theta')$.

2. \textbf{Transition Law Approximation} (Step 2(b)): For the proposed parameter $\theta'$, construct particle-based approximations $\{\mu_{t-1+j\Delta_l,\theta'}^{l,N}\}_{j=1}^{\Delta_l^{-1}}$ of the McKean-Vlasov laws at each time step $t\in\{1,\ldots,T\}$. Initialize the empirical measure as
\begin{equation}
\mu_{t-1,\theta'}^{l,N}(dx) = 
\begin{cases}
\delta_{x_0}(dx), & \text{if } t=1, \\
\frac{1}{N}\sum_{i=1}^N\delta_{X_{t-1}^i}(dx), & \text{if } t>1.
\end{cases}
\label{eq:sing_mu_initialization}
\end{equation}
then evolving $N$ particles through the subintervals of $[t-1,t]$ for $j=1,\ldots,\Delta_l^{-1}$ via the discretized dynamics
\begin{align}
X_{t-1+j\Delta_l}^i &= X_{t-1+(j-1)\Delta_l}^i + a_{\theta'}\left(X_{t-1+(j-1)\Delta_l}^i,\overline{\zeta}_{1,\theta'}(X_{t-1+(j-1)\Delta_l}^i,\mu_{t-1+(j-1)\Delta_l,\theta'}^{l,N})\right)\Delta_l \notag \\
&\quad + b_{\theta'}\left(X_{t-1+(j-1)\Delta_l}^i,\overline{\zeta}_{2,\theta'}(X_{t-1+(j-1)\Delta_l}^i,\mu_{t-1+(j-1)\Delta_l,\theta'}^{l,N})\right)\Delta W_{t-1+j\Delta_l}^i,
\label{eq:particle_update}
\end{align}
where $\Delta W_{t-1+j\Delta_l}^i \stackrel{\text{iid}}{\sim} \mathcal{N}_d(0, \Delta_l I_d)$ is a $d$-dimensional Gaussian random variable with mean $0$ and covariance $\Delta_l I_d$, with $I_d$ being the $d \times d$-dimensional identity matrix. The empirical measure at each substep is defined as
\begin{equation}
\mu_{t-1+(j-1)\Delta_l,\theta'}^{l,N} = \frac{1}{N}\sum_{n=1}^N \delta_{X_{t-1+(j-1)\Delta_l}^{n}}(dx),
\label{eq:mu_empirical}
\end{equation}
and the averaged interactions for $m \in \{1,2\}$ are given by \begin{equation}
\overline{\zeta}_{m,\theta'}(X_{t-1+(j-1)\Delta_l}^i,\mu_{t-1+(j-1)\Delta_l,\theta'}^{l,N}) = \frac{1}{N}\sum_{n=1}^N \zeta_{m,\theta'}(X_{t-1+(j-1)\Delta_l}^i,X_{t-1+(j-1)\Delta_l}^n).
\label{eq:zeta_average}
\end{equation}


3. \textbf{Particle Filter Likelihood Estimation} (Step 2(c)): For a given parameter $\theta' \in \Theta$, conditional on the sequence of approximated empirical measures $\{\mu^{l,N}_{t,\theta'}\}$ generated in Step 2(b), we employ a particle filter with $M$ particles to estimate the marginal likelihood. The procedure operates recursively for $s=1, \dots, T$, where particles are propagated from their resampled ancestors at time $s-1$ using the transition kernel composed of $\Delta_l^{-1}$ Euler-Maruyama steps, defined as
\begin{equation}
\xi_s^m \sim \prod_{j=1}^{\Delta_{l}^{-1}} Q_{\mu_{s-1+(j-1)\Delta_l,\theta'}^{l,N},\theta'}^l(z_{j-1}, dz_j),
\label{eq:transition_product}
\end{equation}
where $z_0 = \xi_{s-1}^{A_{s-1}^m}$ and the dynamics utilize the fixed empirical laws. At each time step, we compute the normalized importance weights
\begin{equation}
r_s^m = \frac{g_{\theta'}(\xi_s^m, y_s)}{\sum_{j=1}^M g_{\theta'}(\xi_j^m, y_s)}, \quad m=1,\dots,M,
\label{eq:Normalized_ISW}
\end{equation}
which are subsequently used to resample the ancestor indices $\{A_s^m\}_{m=1}^M$ for the next iteration. The marginal likelihood estimator is updated recursively via
\begin{equation}
\hat{p}_{\theta'}^{l,M,N}(y_{1:s}) = \hat{p}_{\theta'}^{l,M,N}(y_{1:s-1}) \cdot \frac{1}{M} \sum_{m=1}^M g_{\theta'}(\xi_s^m,y_s),
\label{eq:recursive_likelihood}
\end{equation}
yielding the final estimator constructed as the product of the average unnormalized weights across all time steps:
\begin{equation}
\hat{p}_{\theta'}^{l,M,N}(y_{1:T}) := \prod_{s=1}^T \left\{\frac{1}{M}\sum_{m=1}^M g_{\theta'}(\xi^m_s, y_s)\right\}.
\end{equation}
Finally, a single trajectory sample is generated by tracing back the ancestry of a particle selected from the final weights $\{r_T^m\}$, which serves as the proposed latent trajectory for the Metropolis-Hastings update.

4. \textbf{Metropolis-Hastings Step} (Step 2(d)): This step implements the standard Metropolis-Hastings acceptance-rejection mechanism to determine whether to accept the proposed parameter $\theta'$ or retain the current parameter $\theta^{(k-1)}$. The acceptance probability is computed using the ratio of posterior densities (up to normalizing constants) at the proposed and current parameter values
\begin{equation}
\alpha = \min\left\{1,\; \frac{\hat{p}_{\theta'}^{l,M,N}(y_{1:T})\nu(\theta')r(\theta',\theta^{(k-1)})}{\hat{p}_{\theta^{(k-1)}}^{l,M,N}(y_{1:T})\nu(\theta^{(k-1)})r(\theta^{(k-1)},\theta')}\right\}.
\label{eq:alpha_acceptance}
\end{equation}

\textit{Accept-Reject Decision:} Generate a uniform random variable $U \sim \mathcal{U}[0,1]$, the uniform distribution on $[0,1]$. If $U < \alpha$, accept the proposal: set  $\theta^{(k)} = \theta'$, update the likelihood estimate 
\begin{equation}
\hat{p}_{\theta^{(k)}}^{l,M,N}(y_{1:T}) = \hat{p}_{\theta'}^{l,M,N}(y_{1:T}),
\label{eq:p_hat_equality}
\end{equation}
and store the corresponding trajectory\begin{equation}
(\xi_1^{(k)}, \ldots, \xi_T^{(k)}) = (\xi_1^{m^*}, \ldots, \xi_T^{m^*}),
\label{eq:trajectory_equality}
\end{equation} where the trajectory index $m^*$ is sampled from the normalized weights $\{r_T^m\}_{m=1}^M$ defined as:
\begin{equation}
r_T^m = \frac{g_{\theta'}(\xi_T^m,y_T)}{\sum_{j=1}^M g_{\theta'}(\xi_T^j,y_T)}
\end{equation}
computed at the final time step $T$ in Step 2(c). Otherwise, reject the proposal: retain $\theta^{(k)} = \theta^{(k-1)}$, along with its associated likelihood estimate and trajectory.

This acceptance-rejection procedure ensures that the Markov chain $\{\theta^{(k)},\xi_{1:T}^{(k)}\}_{k=0}^K$ has $\overline{\pi}^{l,N}$ (given in \eqref{eqn:bar_pi}) as its stationary distribution, thereby enabling valid posterior inference despite the use of approximated likelihoods from the particle filter. The algorithm is an adaptation of the original PMCMC framework proposed by \cite{andrieu}. Under mild regularity conditions established therein, the resulting Markov chain produces samples from the approximate posterior distribution $\overline{\pi}^{l,N}$.

\subsection{Multilevel PMCMC Algorithm Details}
\label{app:multilevel_details}

This section provides the complete algorithmic details for the multilevel PMCMC sampling algorithm described in Section \ref{subsec:Multilevel-level PMCMC_D} of the main text.

\textbf{1. Base Level Computation} (Step 2): Execute Algorithm~\ref{alg:PMCMC_mckean_vlasov} at base discretization level $l_{\star} \in \mathbb{N}$ for $K_{l_{\star}}$ iterations with $N_{l_{\star}}$ particles in law approximation to obtain base-level samples $\{\theta^{(k)}_{l_{\star}}, x^{(k)}_{l_{\star},1:T}\}_{k=0}^{K_{l_{\star}}}$ and the corresponding likelihood estimation $\{\hat{p}_{\theta_{l_{\star}}^{(k)}}^{{l_{\star}},M,N_{l_{\star}}}(y_{1:T})\}_{k=1}^{K_{l_{\star}}}$.

\textbf{2. Bi-level Telescoping Difference Approximation} (Step 3): For each level $l \in \{l_{\star}+1,\ldots,L\}$, we independently run the bi-level MCMC to compute the corresponding telescoping difference.

\textbf{2.1. Parameter Proposal} (Step 3(b)(i)): For each level $l$, generate candidate parameter $\theta'$ from proposal kernel $r(\theta^{(k-1)}, \theta')$.

\textbf{2.2. Coupled law approximation} (Step 3(b)(ii)): 
For the proposed parameter $\theta'$, construct synchronized particle-based approximations 
$\{\mu_{t-1+j\Delta_l,\theta'}^{l,N_l}\}_{j=1}^{\Delta_l^{-1}}$ and 
$\{\widetilde{\mu}_{t-1+j\Delta_{l-1},\theta'}^{l-1,N_l}\}_{j=1}^{\Delta_{l-1}^{-1}}$ 
of the McKean--Vlasov laws at consecutive discretization levels. The coupled empirical measures are initialized as 
\begin{align}\label{eq:coupled_init}
&\mu_{t-1,\theta'}^{l,N_l}(dx) = 
\widetilde{\mu}_{t-1,\theta'}^{l-1,N_l}(dx) = \delta_{x_0}(dx),
\quad \text{if } t = 1, \\
&\mu_{t-1,\theta'}^{l,N_l}(dx) 
= \frac{1}{N_l}\sum_{i=1}^{N_l}\delta_{X^{l,i}_{t-1}}(dx) 
\quad\text{and}\quad
\widetilde{\mu}_{t-1,\theta'}^{l-1,N_l}(dx) 
= \frac{1}{N_l}\sum_{i=1}^{N_l}\delta_{\widetilde{X}^{l-1,i}_{t-1}}(dx),
\quad \text{if } t > 1.\nonumber
\end{align}
Subsequently, $N_l$ coupled particle pairs are evolved over the subintervals of $[t-1,t]$ 
using correlated Brownian increments. For $j = 1,\ldots,\Delta_l^{-1}$ and each particle 
$i \in \{1,\ldots,N_l\}$, the fine-level particles evolve according to
\begin{align}
\label{eq:fine_particles}
X_{t-1+j\Delta_l}^{l,i} &= X_{t-1+(j-1)\Delta_l}^{l,i} 
+ a_{\theta'}\left(X_{t-1+(j-1)\Delta_l}^{l,i}, \zeta_{1,\theta'}\left(X_{t-1+(j-1)\Delta_l}^{l,i}, \mu_{t-1+(j-1)\Delta_l,\theta'}^{l,N_l}\right)\right)\Delta_l \notag \\
&\quad + b_{\theta'}\left(X_{t-1+(j-1)\Delta_l}^{l,i}, \zeta_{2,\theta'}\left(X_{t-1+(j-1)\Delta_l}^{l,i}, \mu_{t-1+(j-1)\Delta_l,\theta'}^{l,N_l}\right)\right)\Delta W_{t-1+j\Delta_l}^i
\end{align}
where $\Delta W_{t-1+j\Delta_l}^i \stackrel{\text{iid}}{\sim} \mathcal{N}_d(0,\Delta_l I_d)$. Simultaneously, for $j = 1,\ldots,\Delta_{l-1}^{-1}$, the coarse-level particles evolve as
\begin{align}
\label{eq:coarse_particles}
&\widetilde{X}_{t-1+j\Delta_{l-1}}^{l-1,i}  \notag \\
&= \widetilde{X}_{t-1+(j-1)\Delta_{l-1}}^{l-1,i} \notag \\
&\quad + a_{\theta'}\left(\widetilde{X}_{t-1+(j-1)\Delta_{l-1}}^{l-1,i}, \zeta_{1,\theta'}\left(\widetilde{X}_{t-1+(j-1)\Delta_{l-1}}^{l-1,i}, \widetilde{\mu}_{t-1+(j-1)\Delta_{l-1},\theta'}^{l-1,N_l}\right)\right)\Delta_{l-1} \notag \\
&\quad + b_{\theta'}\left(\widetilde{X}_{t-1+(j-1)\Delta_{l-1}}^{l-1,i}, \zeta_{2,\theta'}\left(\widetilde{X}_{t-1+(j-1)\Delta_{l-1}}^{l-1,i}, \widetilde{\mu}_{t-1+(j-1)\Delta_{l-1},\theta'}^{l-1,N_l}\right)\right)\Delta \widetilde{W}_{t-1+j\Delta_{l-1}}^i
\end{align}
where $\Delta \widetilde{W}_{t-1+j\Delta_{l-1}}^i$ are constructed from the fine-level Brownian increments to ensure coupling via the aggregation
$\Delta \widetilde{W}_{t-1+j\Delta_{l-1}}^i = \Delta W_{t-1+(2j-1)\Delta_l}^i + \Delta W_{t-1+2j\Delta_l}^i$, the empirical measures at each substep are
\begin{align}
\label{eq:empirical_measures}
\mu_{t-1+(j-1)\Delta_l,\theta'}^{l,N_l} &= \frac{1}{N_l}\sum_{n=1}^{N_l}\delta_{X_{t-1+(j-1)\Delta_l}^{l,n}}(dx), \notag \\
\widetilde{\mu}_{t-1+(j-1)\Delta_{l-1},\theta'}^{l-1,N_l} &= \frac{1}{N_l}\sum_{n=1}^{N_l}\delta_{\widetilde{X}_{t-1+(j-1)\Delta_{l-1}}^{l-1,n}}(dx),
\end{align}
and the averaged interactions are given by,  for $m \in \{1,2\}$, 
\begin{align}
\label{eq:averaged_interactions}
\zeta_{m,\theta'}\left(X_{t-1+(j-1)\Delta_l}^{l,i}, \mu_{t-1+(j-1)\Delta_l,\theta'}^{l,N_l}\right) &= \frac{1}{N_l}\sum_{n=1}^{N_l}\zeta_{m,\theta'}\left(X_{t-1+(j-1)\Delta_l}^{l,i}, X_{t-1+(j-1)\Delta_l}^{l,n}\right), \notag \\
\zeta_{m,\theta'}\left(\widetilde{X}_{t-1+(j-1)\Delta_l}^{l,i}, \widetilde{\mu}_{t-1+(j-1)\Delta_l,\theta'}^{l,N_l}\right) &= \frac{1}{N_l}\sum_{n=1}^{N_l}\zeta_{m,\theta'}\left(\widetilde{X}_{t-1+(j-1)\Delta_l}^{l,i}, \widetilde{X}_{t-1+(j-1)\Delta_l}^{l,n}\right).
\end{align}

\textbf{2.3. Delta Particle Filter Likelihood Estimation} (Step 3(b)(iii)): For a given parameter $\theta' \in \Theta$, conditional on the sequence of coupled approximated empirical measures $\{\mu^{l,N_l}_{t,\theta'}, \widetilde{\mu}^{l-1,N_l}_{t,\theta'}\}$ generated in Step 3(b)(ii), we employ a Delta Particle Filter with $M$ coupled particle pairs to estimate the joint likelihood. The procedure operates recursively for $s=1, \dots, T$, where coupled particle pairs are propagated from their resampled ancestors at time $s-1$ using the synchronized coupled transition kernel, denoted as
\begin{equation}
(\xi_s^{m,l}, \widetilde{\xi}_s^{m,l-1}) \sim \check{P}_{\theta'}^l \left( (\xi_{s-1}^{A_{s-1}^m,l}, \widetilde{\xi}_{s-1}^{A_{s-1}^m,l-1}), d(x, \tilde{x}) \right),
\label{eq:coupled_transition_kernel}
\end{equation}
where the dynamics utilize the fixed coupled empirical laws. At each time step, we compute the normalized coupled importance weights
\begin{equation}
r_s^m = \frac{H_{s,\theta'}(\xi_s^{m,l}, \widetilde{\xi}_s^{m,l-1}; y_s)}{\sum_{j=1}^M H_{s,\theta'}(\xi_j^{m,l}, \widetilde{\xi}_j^{m,l-1}; y_s)}, \quad m=1,\dots,M,
\label{eq:coupled_ISW}
\end{equation}
where $H_{s,\theta'}$ is the coupled observation density defined in \eqref{eq:hk_ch}. These weights are subsequently used to resample the ancestor indices $\{A_s^m\}_{m=1}^M$ for the next iteration. The coupled likelihood estimator is updated recursively via
\begin{equation}
\hat{p}^{l,M,N_l}_{\theta'}(y_{1:s}) = \hat{p}^{l,M,N_l}_{\theta'}(y_{1:s-1}) \cdot \frac{1}{M}\sum_{m=1}^M H_{s,\theta'}(\xi_s^{m,l}, \widetilde{\xi}_s^{m,l-1};y_s),
\label{eq:coupledlikelihood_RU}
\end{equation}
yielding the final estimator constructed as the product of the average unnormalized weights across all time steps:
\begin{equation}
\hat{p}^{l,M,N_l}_{\theta'}(y_{1:T}) := \prod_{s=1}^T\left\{\frac{1}{M}\sum_{m=1}^M H_{s,\theta'}(\xi_s^{m,l}, \widetilde{\xi}_s^{m,l-1};y_s)\right\}.
\label{eq:jointF_coupledLikelihood}
\end{equation}
Finally, a single coupled trajectory sample is generated by tracing back the ancestry of a particle pair selected from the final weights $\{r_T^m\}$, which serves as the proposed coupled latent trajectory for the Metropolis-Hastings update.


\textbf{2.4. Metropolis-Hastings Step} (Step 3(b)(iv)): This step implements the standard Metropolis-Hastings acceptance-rejection mechanism to determine whether to accept the proposed parameter $\theta'$ or retain the current parameter $\theta^{(k-1)}$. The acceptance probability is computed using the ratio of posterior densities (up to normalizing constants) at the proposed and current parameter values based on the coupled likelihood estimates:
\begin{align}
    \alpha = \min\left\{1,\; \frac{\hat{p}^{l,M,N_l}_{\theta'}(y_{1:T})\nu(\theta')r(\theta',\theta^{(k-1)})}{\hat{p}^{l,M,N_l}_{\theta^{(k-1)}}(y_{1:T})\nu(\theta^{(k-1)})r(\theta^{(k-1)},\theta')}\right\}.
    \label{eq:alpha}
\end{align}
Accept-Reject Decision: Generate a uniform random variable $U \sim \mathcal{U}[0,1]$. If $U < \alpha$, the proposal is accepted: we set $\theta^{(k)} = \theta'$, update the coupled likelihood estimate 
\begin{equation}
\hat{p}^{l,M,N_l}_{\theta^{(k)}}(y_{1:T}) = \hat{p}^{l,M,N_l}_{\theta'}(y_{1:T}),
\label{eq:likelihood_update}
\end{equation}
and store the corresponding coupled trajectory 
\begin{equation}
    \left((\xi_1^{(k)},\ldots,\xi_T^{(k)}), (\widetilde{\xi}_1^{(k)},\ldots,\widetilde{\xi}_T^{(k)})\right) = \left((\xi_1^{m^*,l},\ldots,\xi_T^{m^*,l}), (\widetilde{\xi}_1^{m^*,l-1},\ldots,\widetilde{\xi}_T^{m^*,l-1})\right),
    \label{eq:trajectory_update}
\end{equation} 
where the trajectory index $m^*$ is resampled from the final normalized weights $\{r_T^m\}_{m=1}^M$ with
\begin{align*}
    r_T^m = \frac{H_{T,\theta'}(\xi_T^{m,l}, \widetilde{\xi}_T^{m,l-1};y_T)}{\sum_{j=1}^M H_{T,\theta'}(\xi_T^{j,l}, \widetilde{\xi}_T^{j,l-1};y_T)}
\end{align*}
computed at the final time step $T$ in Step 3(b)(iii). Otherwise, the proposal is rejected: we retain $\theta^{(k)} = \theta^{(k-1)}$ along with its associated coupled likelihood estimate and trajectory pair.

\textbf{2.5. Importance Weight Correction for Telescoping Difference}: The coupled observation density $H_{k,\theta}(x,x';y_k)$ introduced in Step 3(b)(iii) induces sampling from the coupled measure $\check{\pi}^{l,N_l}$ defined in equation (\ref{eq:main_tar}), which does not directly correspond to the difference of the individual posterior distributions $\bar{\pi}^{l,N_l}$ and $\bar{\pi}^{l-1,N_l}$. To obtain samples from the correct posterior difference, we apply importance weighting correction using the auxiliary weight function
$\check{H}_{k,\theta}(x,x';y_k) = \frac{g_\theta(x,y_k)}{H_{k,\theta}(x,x';y_k)}$,
as defined in Section \ref{subsec: cou_pos_dis}. For the coupled trajectory samples $\{(\theta^l(k), x_{1:T}^{l}(k), \widetilde{x}_{1:T}^{l-1}(k))\}_{k=0}^{K_l}$ generated by the bi-level MCMC, the telescoping difference estimator in equation (\ref{eq:mcmc_bl_est}) employs these correction weights to properly account for the distributional mismatch. Specifically, the fine-level expectation is weighted by $\prod_{s=1}^T \check{H}_{s,\theta^{(k)}_l}(x_s^{l}(k), \widetilde{x}_s^{l-1}(k);y_s)$ when evaluating $\varphi(\theta^{(k)}_l, x_{1:T}^{l}(k))$, while the coarse-level expectation is weighted by $\prod_{s=1}^T \check{H}_{s,\theta^{(k)}_l}(\widetilde{x}_s^{l-1}(k), x_s^{l}(k);y_s)$ when evaluating $\varphi(\theta^{(k)}_l, \widetilde{x}_{1:T}^{l-1}(k))$. This change of measure technique, rigorously justified through equation (\ref{eq:main_eq}), ensures that the normalized weighted empirical averages converge to the true posterior difference $\bar{\pi}^{l,N_l}(\varphi) - \bar{\pi}^{l-1,N_l}(\varphi)$ as $K_l \to \infty$, thereby maintaining the validity of the MLMC framework while enabling effective variance reduction through the coupled sampling mechanism.

\subsection{Complete Algorithm Pseudocode}
\label{app:algorithms}
\renewcommand{\thealgorithm}{A.\arabic{algorithm}}
\setcounter{algorithm}{0}

{\small
\begin{algorithm}[H]
	\begin{enumerate}
        \item{\textbf{Initialization:} Input discretization level $l\in\mathbb{N}_0$ (set $\Delta_l = 2^{-l}$), particle counts $N\in\mathbb{N}$ (for transition law approximation) and $M\in\mathbb{N}$ (for filtering), initial parameter $\theta^{(0)}\in\Theta$, MCMC iterations $K\in\mathbb{N}$, and data $y_{1:T}$.}
        
		\item{\textbf{MCMC Iteration Loop:} For $k=1,\ldots,K$:
			
			\begin{enumerate}[(a)]
				\item \textbf{Parameter Proposal:} Generate $\theta'$ from proposal kernel $r(\theta^{(k-1)},\theta')d\theta'$.
				
				\item \textbf{Law Approximation:} For each time step $t\in\{1,\ldots,T\}$:
				\begin{enumerate}[i.]
					\item Initialize empirical measure according to Eq.\eqref{eq:sing_mu_initialization}.
					
					\item For $j=1,\ldots,\Delta_l^{-1}$ and each particle $i\in\{1,\ldots,N\}$, evolve particles $X_{t-1+j\Delta_l}^i$ via Eq.\eqref{eq:particle_update}, using the interaction terms and empirical measures defined in Eqs. \eqref{eq:zeta_average}and \eqref{eq:mu_empirical}. The increments $\Delta W_{t-1+j\Delta_l}^i \stackrel{\text{iid}}{\sim} \mathcal{N}_d(0,\Delta_l I_d)$.
                    
					\item Store the approximated laws $\{\mu_{t-1+j\Delta_l,\theta'}^{l,N}\}_{j=1}^{\Delta_l^{-1}}$.
					
				\end{enumerate}

                \item \textbf{Particle Filter Likelihood Estimation:}
                \begin{enumerate}[i.]
                    \item Initialize $M$ particles $\{\xi_1^m\}_{m=1}^M$ from transition kernels using Eq.\eqref{eq:transition_product}  with $s = 1$. Set ancestor indices $A_0^m = m$ and marginal likelihood $\hat{p}_{\theta'}^{l,M,N}(y_0) = 1$.
                    
                    \item For $s=1,\ldots,T$:
                    \begin{itemize}
                        \item \textbf{Resampling:} Compute normalized weights via Eq.\eqref{eq:Normalized_ISW} and sample ancestor indices $A_s^m$ according to $\{r_s^m\}_{m=1}^M$.
                        \item \textbf{Likelihood Update:} Update likelihood by Eq.\eqref{eq:recursive_likelihood}.                 
                        \item \textbf{Propagation:} If $s < T$, for $m\in \{1,\cdots,M\},$ sample $\xi_{s+1}^m|\xi_{s}^{A_{s}^m}$ from transition kernel
                        using Eq.\eqref{eq:transition_product} with $s = s+1$, 
                        where $x_{s}^m = \xi_{s}^{A_{s}^m}$ and $(\xi_1^m,\cdots,\xi_{s+1}^m) = (\xi_1^{A_s^m},\cdots,\xi_{s}^{A_s^m}, \xi_{s+1}^m).$
                    \end{itemize}
                    \item Sample final trajectory index $m^*$ from final weights $\{r_T^m\}_{m=1}^M$.
                \end{enumerate}
                
                \item \textbf{Metropolis-Hastings Step:} 
                \begin{enumerate}[i.]
                    \item Compute acceptance probability according to Eq.\eqref{eq:alpha_acceptance}.
                    \item Generate $U\sim\mathcal{U}[0,1]$. If $U < \alpha$, accept: $\theta^{(k)} = \theta'$, and update the likelihood and trajectory according to Eq.\eqref{eq:p_hat_equality} and \eqref{eq:trajectory_equality}. Otherwise, reject: $\theta^{(k)} = \theta^{(k-1)}$ with unchanged likelihood and trajectory. 
                \end{enumerate}
			\end{enumerate}
	}
    \item{\textbf{Output:} Return parameter sample chain $\{\theta^{(k)}\}_{k=1}^K$, corresponding likelihood estimates $\{\hat{p}_{\theta^{(k)}}^{l,M,N}(y_{1:T})\}_{k=1}^K$, and trajectory samples $\{(\xi_1^{(k)},\ldots,\xi_T^{(k)})\}_{k=1}^K$.}
\end{enumerate}
\caption{PMCMC Algorithm for MVSDE Parameter Estimation}
\label{alg:PMCMC_mckean_vlasov}
\end{algorithm}
}

{\tiny
\begin{algorithm}[H]
\setlength{\itemsep}{1.5 em}  
\setlength{\parskip}{1 em}
\begin{enumerate}
\item{\textbf{Initialization:} Given base level $l^* \in \mathbb{N}_0$ and finest level $L \in \mathbb{N}$ ($L > l^*$) with $\Delta_l = 2^{-l}$. Set particle counts $N_{l^*} \in \mathbb{N}$ for base law approximation, $\{N_{l^*+1}, \ldots, N_L\} \subset \mathbb{N}$ for coupled law approximation, and $M \in \mathbb{N}$ for filtering. Set iterations $K_{l^*} \in \mathbb{N}$ for base MCMC and $\{K_{l^*+1}, \ldots, K_L\} \subset \mathbb{N}$ for bi-level MCMC. Input initial parameter $\theta^{(0)} \in \Theta$ and data $y_{1:T}$.}

\item{\textbf{Base Level Computation:} Execute Algorithm \ref{alg:PMCMC_mckean_vlasov} at level $l^*$ for $K_{l^*}$ iterations with $N_{l^*}$ particles in Step 2(b) Law Approximation. Obtain base level samples $\{\theta_{l^*}^{(k)}, x_{l^*,1:T}^{(k)}\}_{k=0}^{K_{l^*}}$ and likelihood estimates  $\{\hat{p}_{\theta_{l_{\star}}^{(k)}}^{l^*,M,N_{l^*}}(y_{1:T})\}_{k=0}^{K_{l^*}}$.}

\item{\textbf{Telescoping Difference Computation:} For each discretization level $l \in \{l^* + 1, \ldots, L\}$:
\begin{enumerate}[(a)]
\item \textbf{Bi-level MCMC Setup:} Set $\Delta_l = 2^{-l}$, $\Delta_{l-1} = 2^{-(l-1)}$. Initialize MCMC counter $k = 0$.

\item \textbf{MCMC Iteration Loop:} For $k = 1, \ldots, K_l$:

\begin{enumerate}[i.]
\item \textbf{Parameter Proposal:} Generate $\theta'$ from kernel $r(\theta^{(k-1)}, \theta') d\theta'$.

\item \textbf{Coupled Law Approximation:} For each time step $t \in \{1, \ldots, T\}$:
\begin{enumerate}[1.]
\item Initialize coupled empirical measures via Eq.~\eqref{eq:coupled_init}.

\item For $j = 1, \ldots, \Delta_l^{-1}$ and particle $i \in \{1, \ldots, N_l\}$, evolve fine-level particles via Eq.~\eqref{eq:fine_particles}, using the interaction terms and empirical measures defined in Eqs.~\eqref{eq:averaged_interactions} and \eqref{eq:empirical_measures}. The increments $\Delta W_{t-1+j\Delta_l}^i$ are i.i.d.~random vectors distributed as $\mathcal{N}_d(0,\Delta_l I_d)$.

\item Analogously, for $j = 1, \ldots, \Delta_{l-1}^{-1}$, evolve coarse-level particles via Eq.~\eqref{eq:coarse_particles} and adopts the interaction terms and empirical measures via Eqs.~\eqref{eq:averaged_interactions} and \eqref{eq:empirical_measures}, respectively. The coarse-level increments $\Delta \widetilde{W}_{t-1+j\Delta_{l-1}}^i$ are constructed from the fine-level increments.

\item Store the approximated coupled laws
$$\left\{\left(\mu_{t-1+\Delta_l,\theta'}^{l,N_l}, \widetilde{\mu}_{t-1+\Delta_{l-1},\theta'}^{l-1,N_l}\right), \ldots, \left(\mu_{t,\theta'}^{l,N_l}, \widetilde{\mu}_{t,\theta'}^{l-1,N_l}\right)\right\}.$$
\end{enumerate}
\end{enumerate}
\end{enumerate}
}
\end{enumerate}
\caption{Multilevel PMCMC for McKean-Vlasov SDE - Part I }
\label{MLPMCMC_mckean_vlasov_part1}
\end{algorithm}
}

{\tiny
\addtocounter{algorithm}{-1}  
\begin{algorithm}[H]
\begin{enumerate}
\setcounter{enumi}{2}
\item{\textbf{Telescoping Difference Computation (continued):} For each discretization level $l \in \{l^* + 1, \ldots, L\}$:

\begin{enumerate}[(a)]
\setcounter{enumii}{1}
\item \textbf{MCMC Iteration Loop (continued):} For $k = 1, \ldots, K_l$:

\begin{enumerate}[i.]
\setcounter{enumiii}{2}
\item \textbf{Delta Particle Filter Likelihood Estimation:}
\begin{enumerate}[i.]
\item Initialize $M$ coupled particle pairs $\left\{\left(\xi_1^{m,l}, \widetilde{\xi}_1^{m,l-1}\right)\right\}_{m=1}^M$ from coupled transition kernels using $\check{P}_{\theta'}^l\left(\left(x_{0}, \widetilde{x}_{0}\right), d\left(u_1^l, \widetilde{u}_1^{l-1}\right)\right)$. Set $A_0^m = m$ and $p_{\theta'}^{l,M,N_l}(y_0) = 1$.

\item For $s = 1, \ldots, T$:
\begin{itemize}
\item Compute coupled normalized weights via Eq.\eqref{eq:coupled_ISW}, and sample $A_s^m$ from $\left\{r_s^m\right\}_{m=1}^M$.
\item Update likelihood according to Eq.\eqref{eq:coupledlikelihood_RU}.
\item If $s < T$, sample $\xi_{s+1}^{m,l}, \widetilde{\xi}_{s+1}^{m,l-1}|\xi_{s}^{A_s^m,l}, \widetilde{\xi}_{s}^{A_s^m,l-1}$ from the coupled transition kernels by Eq.\eqref{eq:coupled_transition_kernel} with $s = s+1$, where the input measures have been generated in Step (b)(ii). For $m \in \{1, \ldots, M\}$, $\left(\xi_1^{m,l}, \ldots, \xi_{s+1}^{m,l}\right) = \left(\xi_1^{A_s^m,l}, \ldots, \xi_s^{A_s^m,l}, \xi_{s+1}^{m,l}\right),$
$\left(\widetilde{\xi}_1^{m,l-1}, \ldots, \widetilde{\xi}_{s+1}^{m,l-1}\right) = \left(\widetilde{\xi}_1^{A_s^m,l-1}, \ldots, \widetilde{\xi}_s^{A_s^m,l-1}, \widetilde{\xi}_{s+1}^{m,l-1}\right).$
\end{itemize}

\item Sample final coupled trajectory index $m^*$ from final weights $\left\{r_T^m\right\}_{m=1}^M$.
\end{enumerate}

\item \textbf{Metropolis-Hastings Step:}
\begin{enumerate}[i.]
\item Compute acceptance probability $\alpha$ according to Eq. \eqref{eq:alpha}:
\item Generate $U \sim \mathcal{U}[0,1]$. If $U < \alpha$, accept: $\theta^{(k)} = \theta'$,  and  update the likelihood and trajectory according to Eq.\eqref{eq:likelihood_update} and Eq.\eqref{eq:trajectory_update}.
Otherwise, reject: $\theta^{(k)} = \theta^{(k-1)}$ with unchanged likelihood and trajectory.
\end{enumerate}
\end{enumerate}
\end{enumerate}
}

\item{\textbf{Output:} 
\begin{itemize}
\item Base level samples: $\{\theta_{l^*}^{(k)}\}_{k=0}^{K_{l^*}}$, $\{\hat{p}_{\theta^{(k)}}^{l^*,M,N_{l^*}}(y_{1:T})\}_{k=0}^{K_{l^*}}$, $\{(\xi_1^{(k)}, \ldots, \xi_T^{(k)})\}_{k=0}^{K_{l^*}}$.

\item Bi-level coupled samples: For each level $l \in \{l^* + 1, \ldots, L\}$, return $\{\theta_{l}^{(k)}\}_{k=0}^{K_{l}}$, $\{\hat{p}_{\theta^{(k)}}^{l,M,N_{l}}(y_{1:T})\}_{k=0}^{K_{l}}$, $\left\{\left(\xi_1^{(k)}, \ldots, \xi_T^{(k)}\right), \left(\widetilde{\xi}_1^{(k)}, \ldots, \widetilde{\xi}_T^{(k)}\right)\right\}_{k=0}^{K_l}$.
\end{itemize}
}

\end{enumerate}
\caption{Multilevel PMCMC for McKean-Vlasov SDE - Part II}
\label{alg:MLPMCMC_mckean_vlasov_part2}
\end{algorithm}
}

\subsection{3D Neuron Model: Detailed Description}
\label{app:neuron_model_details}
The three-dimensional state vector $x = (x_1, x_2, x_3)$ represents distinct neuronal variables: $x_1$ corresponds to the membrane potential $V$, $x_2$ denotes the recovery variable $w$, and $x_3$ characterizes the synaptic gating variable $y$. The drift term $a(t, x, \mu)$ governs the deterministic dynamics through three coupled components. The first component models the FitzHugh-Nagumo neuronal dynamics with external input $I$ and mean-field synaptic coupling modulated by the reversal potential $V_{rev}$ and coupling strength $J$. The second component captures the recovery dynamics with time scale parameter $c$ and coupling coefficients $a$ and $b$. The third component describes synaptic gating kinetics, incorporating rise and decay rates $a_r$ and $a_d$, maximum transmission $T_{max}$, and a sigmoid activation function controlled by the steepness parameter $\lambda$ and threshold potential $V_T$.
The diffusion matrix $b(t, x, \mu)$ introduces stochasticity into the system. External noise with intensity $b_{ext}$ directly affects the membrane potential, while synaptic noise with strength $b_J$ contributes through the mean-field term. The recovery variable experiences no direct noise. The synaptic gating variable incorporates state-dependent noise through the function $b_{32}(x)$, which combines the synaptic dynamics with an exponential damping factor controlled by parameters $\Gamma$ and $\Lambda$, and includes an indicator function ensuring $x_3$ remains within the physically meaningful interval $(0,1)$. 

\subsection{Supplementary Tables and Figures}
\label{app:tables_figures}

\setcounter{table}{0}
\setcounter{figure}{0}

\begin{table}[H]
\centering
\begin{tabular}{|c|c|c|c|c|c|c|}
\hline
$V_0 = 0$ & $\sigma_{V_0} = 0.4$ & $a = 0.7$ & $b = 0.8$ & $c = 0.08$ & $I = 0.5$ & $b_{ext} = 0.5$ \\
\hline
$w_0 = 0.5$ & $\sigma_{w_0} = 0.4$ & $V_{rev} = 1$ & $a_r = 1$ & $a_d = 1$ & $T_{max} = 1$ & $\lambda = 0.2$ \\
\hline
$y_0 = 0.3$ & $\sigma_{y_0} = 0.05$ & $J = 1$ & $b_J = 0.2$ & $V_T = 2$ & $\Gamma = 0.1$ & $\Lambda = 0.5$ \\
\hline
\end{tabular}
\caption{Model Parameters and Their Values}
\label{tab:model_parameters}
\end{table}


\vspace{2em}

\begin{table}[H]
	\centering
	\begin{tabular}{c|c|c}
		Parameter & PMCMC & MLPMCMC \\
		\hline
		$I$ & -3.32 & -3.10 \\
		$J$ & -3.39 & -2.99 \\
		$\log(c)$ & -3.92 & -2.18 \\
		$\log(\lambda)$ & -4.50 & -2.76 \\
		$\log(b_{ext})$ & -3.48 & -2.83 \\
		$\log(\Gamma)$ & -3.69 & -2.47 \\
		$\log(\sigma_1)$ & -3.74 & -2.38\\
        $\log(\sigma_2)$ & -3.64 & -3.26\\
        $\log(\sigma_3)$ & -3.24 & -3.17\\
	\end{tabular}
    \caption{Comparison of Estimated Convergence Rates}
	\label{tab:rate_3d}
\end{table}

\vspace{2em}

\begin{center}
\includegraphics[width=0.9\textwidth]{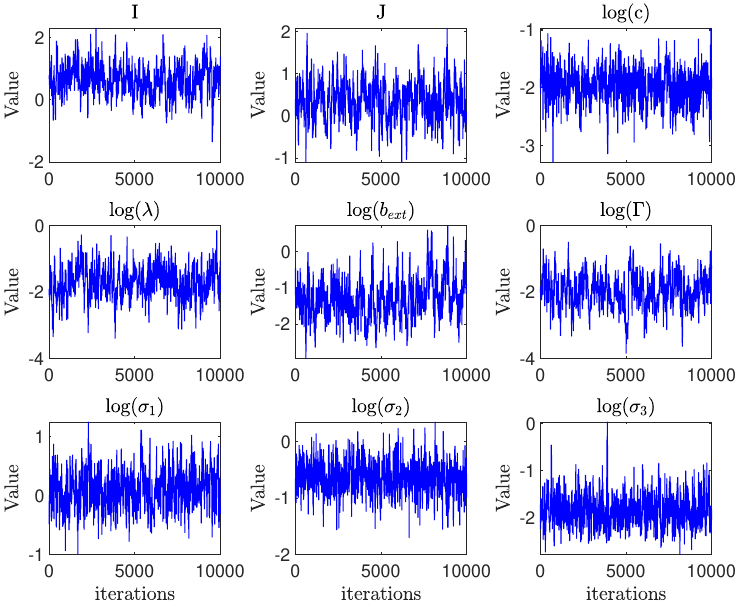}
\captionof{figure}{Trace plots for the estimated parameters of the 3D Neuron Model using PMCMC.}
\label{fig:trace_d3_pmcmc}
\end{center}

\vspace{2em}

\begin{center}
\includegraphics[width=0.9\textwidth]{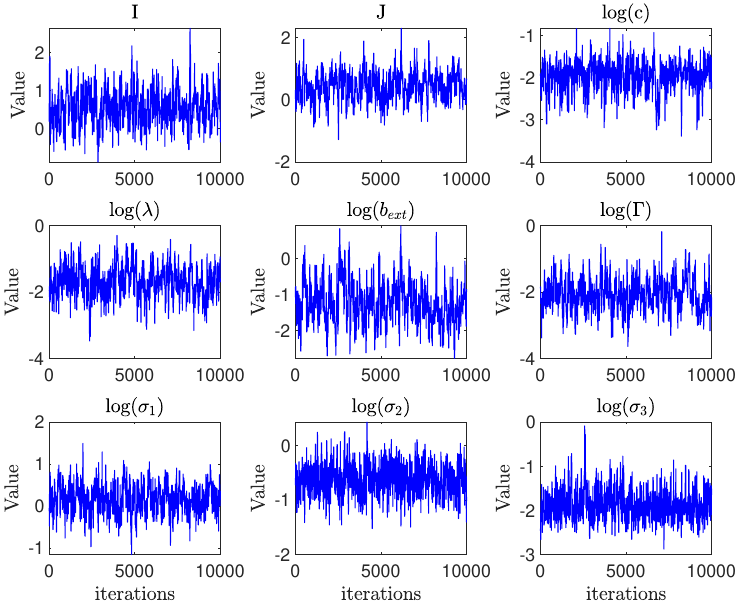}
\captionof{figure}{Trace plots for the estimated parameters of the 3D Neuron Model using MLPMCMC.}
\label{fig:Trace_d3_ml}
\end{center}

\vspace{2em}

\begin{center}
\includegraphics[width=0.9\textwidth]{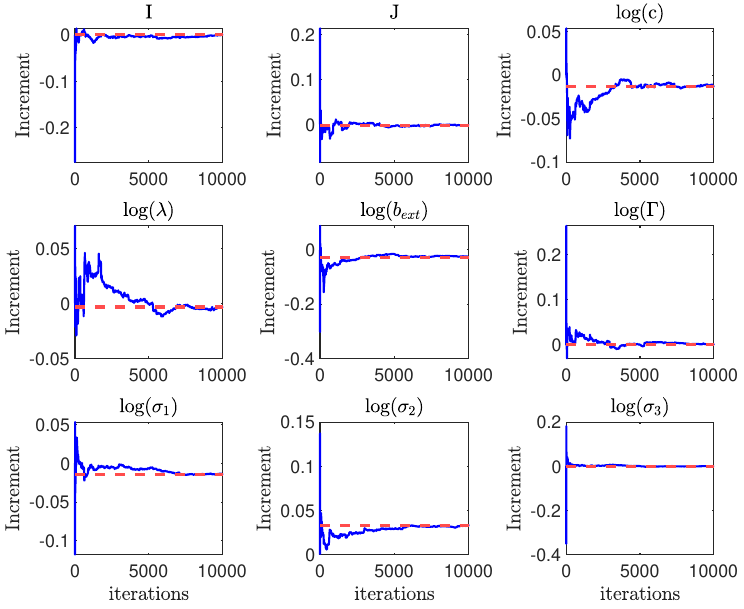}
\captionof{figure}{Running mean of the increments for the estimated parameters of the 3D Neuron Model using MLPMCMC.}
\label{fig:incre_d3_ml2}
\end{center}

\vspace{2em}

\subsection{Regularity Assumptions}
\label{app:assumptions}

Let $\mathcal{P}(\mathbb{R}^d)$ denote the set of probability measures on the measurable space $(\mathbb{R}^d, \mathcal{B}(\mathbb{R}^d))$, where $\mathcal{B}(\mathbb{R}^d)$ represents the Borel $\sigma$-field on $\mathbb{R}^d$. 
We denote by $\mathcal{C}_b^k(\mathbb{R}^{d_1}, \mathbb{R}^{d_2})$ the space of functions $f: \mathbb{R}^{d_1}\rightarrow \mathbb{R}^{d_2}$ with bounded derivatives up to order $k$, and $\mathcal{B}_b(\mathbb{R}^{d_1}, \mathbb{R}^{d_2})$ for the collection of bounded measurable functions from $\mathbb{R}^{d_1}$ to $\mathbb{R}^{d_2}$. We assume that all Markov chains are initialized from their respective stationary distributions and enforce the following conditions in the theoretical analysis.

\begin{hypA}\label{ass:1}
For each $(\theta,\mu)\in\Theta\times\mathcal{P}(\mathbb{R}^d)$ and a fixed $x\in \mathbb{R}$, we have $(a_{\theta}(\cdot,\overline{\zeta}_{1,\theta}(\cdot,\mu)),\zeta_{1,\theta}(\cdot,x))\in\mathcal{C}_b^2(\mathbb{R}^{d+1},\mathbb{R}^d)\cap\mathcal{B}_b(\mathbb{R}^{d+1},\mathbb{R}^d)\times\mathcal{C}_b^2(\mathbb{R}^{2d},\mathbb{R})\cap\mathcal{B}_b(\mathbb{R}^{2d},\mathbb{R})$ and $(b_{\theta}(\cdot, \overline{\zeta}_{2,\theta}(\cdot, \mu)),\zeta_{2,\theta}(\cdot,x)) \in \mathcal{C}_b^2(\mathbb{R}^{d+1}, \mathbb{R}^{d \times d}) \cap \mathcal{B}_b(\mathbb{R}^{d+1}, \mathbb{R}^{d \times d})\times\mathcal{C}_b^2(\mathbb{R}^{2d},\mathbb{R})\cap\mathcal{B}_b(\mathbb{R}^{2d},\mathbb{R})$. All upper bounds of these functionals hold uniformly over $\theta\in\Theta$.
Furthermore, there exists $k > 0$ such that the uniform ellipticity condition holds:
$$\inf_{(x,y) \in \mathbb{R}^{d}\times \mathbb{R}} \inf_{v \in \mathbb{R}^d \setminus \{0\}} \frac{v^{\top} b_{\theta}(x,y)^{\top} b_{\theta}(x,y) v}{\|v\|^2} \geq k, \qquad \forall \theta \in \Theta.$$
\end{hypA}

\begin{hypA}\label{ass:2}
For every $y\in \mathsf{Y}$, the function $(\theta,x)\mapsto g_{\theta}(x,y)$ belongs to $\mathcal{C}_b^2(\Theta\times\mathbb{R}^d, \mathbb{R})\cap\mathcal{B}_b(\Theta\times\mathbb{R}^d, \mathbb{R})$. Furthermore, $\inf_{(\theta,x,y)\in\Theta\times\mathbb{R}^{d}\times\mathsf{Y}}g_{\theta}(x,y)>0$.
\end{hypA}

\begin{hypA}\label{ass:3}
For each $l\in\mathbb{N}$, the Markov kernels $R_l$ (arising from Algorithm \ref{alg:PMCMC_mckean_vlasov}) and $\check{R}_l$ (arising from Algorithm \ref{MLPMCMC_mckean_vlasov_part1}) admit invariant measures with respect to which they are reversible. Moreover, they are uniformly ergodic with 1-step minorization conditions that are uniform in $l$.
\end{hypA}

\subsection{Proof of Theorem \ref{thm:main_convergence}}
\label{app:main_proof}

\begin{proof}
The proof proceeds in the following $4$ steps.

\noindent\textbf{Step 1: Error Decomposition.}
For the difference
$\overline{\pi}^{l,N_l,K_l}(\varphi) - \overline{\pi}^{l-1,N_{l},K_l}(\varphi)$ defined in \eqref{eq:mcmc_bl_est} and the distribution
$\overline{\pi}^l$ defined in \eqref{eqn:overline_pi},
we denote $$\Delta_l^{N_l,K_l}[\varphi] := \overline{\pi}^{l,N_l,K_l}[\varphi] - \overline{\pi}^{l-1,N_l,K_l}[\varphi] \quad \text{and}\quad \Delta_l[\varphi] := \overline{\pi}^{l}[\varphi] - \overline{\pi}^{l-1}[\varphi].$$
For $\widehat{\pi}$ defined in \eqref{eq:mlmc_final_estimator}, 
by telescoping from level $l_{\star}$ to $L$ and applying Cauchy–Schwarz inequality, we obtain
\begin{align*}
\widehat{\pi}(\varphi) - \pi(\varphi) &= \left[\overline{\pi}^{l_{\star},N_{l_{\star}},K_{l_{\star}}}[\varphi] - \overline{\pi}^{l_{\star}}[\varphi]\right]
+ \sum_{l=l_{\star}+1}^L\left[\Delta_l^{N_l,K_l}[\varphi] - \Delta_l[\varphi]\right]
+ \left[\overline{\pi}^{L}[\varphi] - \pi[\varphi]\right]\\
&=:\mathcal{T}_1+\mathcal{T}_2+\mathcal{T}_3.
\end{align*}
Taking square and then expectation of the above equation, by means of the Cauchy–Schwarz inequality $\mathbb{E}[(\sum_{i=1}^n X_i)^2] \leq n\sum_{i=1}^n \mathbb{E}[X_i^2]$ with $n=3$, yields
\begin{align}
\label{eqn:3termsum}
\mathbb{E}\left[\left(\widehat{\pi}(\varphi) - \pi(\varphi)\right)^2\right] \leq 3\left\{\mathbb{E}[\mathcal{T}_1^2] + \mathbb{E}[\mathcal{T}_2^2] + \mathcal{T}_3^2\right\}.
\end{align}

\noindent\textbf{Step 2: Bounding $\mathcal{T}_1$.}
We decompose the term $\mathcal{T}_1$ into two distinct components: the stochastic error arising from the MCMC sampling and the systematic bias introduced by the particle filter approximation. By introducing the intermediate quantity $\overline{\pi}^{l_{\star},N_{l_{\star}}}[\varphi]$, which represents the exact expectation with respect to the particle-approximated posterior (i.e., the limit as $K_{l_*} \to \infty$), we write
\begin{align}
\mathcal{T}_1 &= \overline{\pi}^{l_{\star},N_{l_{\star}},K_{l_{\star}}}[\varphi] - \overline{\pi}^{l_{\star}}[\varphi] \nonumber \\
&= \underbrace{\left[\overline{\pi}^{l_{\star},N_{l_{\star}},K_{l_{\star}}}[\varphi] - \overline{\pi}^{l_{\star},N_{l_{\star}}}[\varphi]\right]}_{\text{MCMC Error}} + \underbrace{\left[\overline{\pi}^{l_{\star},N_{l_{\star}}}[\varphi] - \overline{\pi}^{l_{\star}}[\varphi]\right]}_{\text{Particle Approximation Error}}.
\end{align}
Squaring both sides and taking the expectation, we apply the inequality $(a+b)^2 \leq 2a^2 + 2b^2$ to separate the contributions:
\begin{align}
\mathbb{E}[\mathcal{T}_1^2] \leq 2\mathbb{E}\left[\left(\overline{\pi}^{l_{\star},N_{l_{\star}},K_{l_{\star}}}[\varphi] - \overline{\pi}^{l_{\star},N_{l_{\star}}}[\varphi]\right)^2\right] + 2\left(\overline{\pi}^{l_{\star},N_{l_{\star}}}[\varphi] - \overline{\pi}^{l_{\star}}[\varphi]\right)^2.
\end{align}
We bound these two terms separately. 
Firstly, applying the uniform ergodicity condition (Assumption A\ref{ass:3}) and the Corollary 2.1 in \cite{roberts1997geometric}, the variance of the MCMC estimator decays linearly with the number of iterations $K_{l_{\star}}$. Thus, there exists a constant $C$ such that:
    \begin{align}    \mathbb{E}\left[\left(\overline{\pi}^{l_{\star},N_{l_{\star}},K_{l_{\star}}}[\varphi] - \overline{\pi}^{l_{\star},N_{l_{\star}}}[\varphi]\right)^2\right] \leq \frac{C}{K_{l_{\star}}+1}.
    \end{align}
Secondly, according to Lemma \ref{lem:lem2}, the bias induced by the particle approximation with $N_{l_{\star}}$ particles satisfies $|\overline{\pi}^{l_{\star},N_{l_{\star}}}[\varphi] - \overline{\pi}^{l_{\star}}[\varphi]| \leq C N_{l_{\star}}^{-1/2}$. Squaring this yields
    \begin{align}
    \left(\overline{\pi}^{l_{\star},N_{l_{\star}}}[\varphi] - \overline{\pi}^{l_{\star}}[\varphi]\right)^2 \leq \frac{C}{N_{l_{\star}}}.
    \end{align}
Combining these estimates yields the final bound for the base level error:
\begin{align}
\label{eqn:3termsum_1}
\mathbb{E}[\mathcal{T}_1^2] \leq C\left(\frac{1}{K_{l_{\star}}+1} + \frac{1}{N_{l_{\star}}}\right) = C\left(\varepsilon_{\mathrm{MCMC}}^{(l_{\star})} + \varepsilon_{\mathrm{part}}^{(l_{\star})}\right).
\end{align}

\noindent\textbf{Step 3: Bounding $\mathcal{T}_2$.}
We first define the intermediate quantity $$\Delta_l^{N_l}[\varphi] := \overline{\pi}^{l,N_l}[\varphi] - \overline{\pi}^{l-1,N_l}[\varphi],$$ which represents the difference between the exact expectations under the particle-approximated posteriors at levels $l$ and $l-1$, where $\overline{\pi}^{l,N_l}$ and $\overline{\pi}^{l-1,N_l}$ are the distributions we mentioned in \eqref{eq:diff}. We then decompose the total sum into a stochastic MCMC error component and a deterministic particle bias component:
\begin{align*}
\mathcal{T}_2 = \sum_{l=l_{\star}+1}^L \left( \Delta_l^{N_l,K_l}[\varphi] - \Delta_l^{N_l}[\varphi] \right) + \sum_{l=l_{\star}+1}^L \left( \Delta_l^{N_l}[\varphi] - \Delta_l[\varphi] \right) := \mathcal{T}_{2, \text{MCMC}} + \mathcal{T}_{2, \text{part}}.
\end{align*}
Using the inequality $(a+b)^2 \leq 2a^2 + 2b^2$, we have
\begin{equation*}
\mathbb{E}[\mathcal{T}_2^2] \leq 2\mathbb{E}[\mathcal{T}_{2, \text{MCMC}}^2] + 2\mathcal{T}_{2, \text{part}}^2.
\end{equation*}
The term $\mathcal{T}_{2, \text{part}}$ represents the accumulated particle approximation bias across levels. By Lemma \ref{lem:lem1},
\begin{equation*}
|\mathcal{T}_{2, \text{part}}| = \left| \sum_{l=l_{\star}+1}^L \left( \Delta_l^{N_l}[\varphi] - \Delta_l[\varphi] \right) \right| \leq C \sum_{l=l_{\star}+1}^L \frac{\sqrt{\Delta_l}}{\sqrt{N_l}}.
\end{equation*}
Squaring this result yields the multilevel particle error bound
\begin{equation*}
\mathcal{T}_{2, \text{part}}^2 \leq C \left( \sum_{l=l_{\star}+1}^L \frac{\sqrt{\Delta_l}}{\sqrt{N_l}} \right)^2 = C \varepsilon_{\mathrm{part}}^{(\mathrm{ML})}.
\end{equation*}
We next expand the square of the MCMC error sum $\mathcal{T}_{2, \text{MCMC}}$ into diagonal (variance) and off-diagonal (covariance) terms:
\begin{align*}
\mathbb{E}[\mathcal{T}_{2, \text{MCMC}}^2] &= \sum_{l=l_{\star}+1}^L \mathbb{E}\left[ \left( \Delta_l^{N_l,K_l}[\varphi] - \Delta_l^{N_l}[\varphi] \right)^2 \right] \\
&\quad + \sum_{\substack{l,q=l_{\star}+1\\l\neq q}}^L \mathbb{E}\left[ \left( \Delta_l^{N_l,K_l}[\varphi] - \Delta_l^{N_l}[\varphi] \right) \left( \Delta_q^{N_q,K_q}[\varphi] - \Delta_q^{N_q}[\varphi] \right) \right]\\
&=:\mathcal{T}_{2, \text{MCMC}}^{(1)}+\mathcal{T}_{2, \text{MCMC}}^{(2)}.
\end{align*}
To bound $\mathcal{T}_{2, \text{MCMC}}^{(1)}$, we apply Lemma \ref{lem:lem4}. Under stationarity (Assumption \ref{ass:3}), the MCMC error has zero mean, and Lemma \ref{lem:lem4} bounds the variance of the coupled estimator by $O(\Delta_l(K_l+1)^{-1})$. Summing over levels:
\begin{equation*}
\sum_{l=l_{\star}+1}^L \mathbb{E}\left[ \left( \Delta_l^{N_l,K_l}[\varphi] - \Delta_l^{N_l}[\varphi] \right)^2 \right] \leq C \sum_{l=l_{\star}+1}^L \frac{\Delta_l}{K_l+1} = C \varepsilon_{\mathrm{MCMC}}^{(\mathrm{ML})}.
\end{equation*}
To bound $\mathcal{T}_{2, \text{MCMC}}^{(2)}$, Lemma \ref{lem:lem5} gives
\begin{align*}
\sum_{\substack{l,q=l_{\star}+1\\l\neq q}}^L \mathbb{E}\left[ \left( \Delta_l^{N_l,K_l}[\varphi] - \Delta_l^{N_l}[\varphi] \right) \left( \Delta_q^{N_q,K_q}[\varphi] - \Delta_q^{N_q}[\varphi] \right) \right] &\leq C \sum_{\substack{l,q=l_{\star}+1\\l\neq q}}^L \frac{\sqrt{\Delta_l}}{K_l+1} \frac{\sqrt{\Delta_q}}{K_q+1}\\
&= C\varepsilon_{\mathrm{cross}}.
\end{align*}
The upper bound for $\mathbb{E}[\mathcal{T}_{2, \text{MCMC}}^2]$ follows by combining the two inequalities above, yielding $C(\varepsilon_{\mathrm{MCMC}}^{(\mathrm{ML})}+ \varepsilon_{\mathrm{cross}})$.

\noindent\textbf{Step 4: Bounding $\mathcal{T}_3$.}
For $\mathcal{T}_3$, the standard weak convergence analysis under Assumption A\ref{ass:1}, together with Lemma A.12 of \cite{po_mv}, yields
\begin{align}
\label{eqn:3termsum_3}
\mathcal{T}_3^2 = (\overline{\pi}^{L}[\varphi] - \pi[\varphi])^2 \leq C(\Delta_L^2) = C(\varepsilon_{\mathrm{disc}}).
\end{align}
 Substituting the upper bound for $\mathbb{E}[\mathcal{T}_{2, \text{MCMC}}^2]$ achieved in the last step, together with \eqref{eqn:3termsum_1} and \eqref{eqn:3termsum_3}, into \eqref{eqn:3termsum} completes the proof.
\end{proof}

\subsection{Technical Proofs}\label{app:appendix}

\begin{lem}\label{lem:lem1}
Under Assumptions (A\ref{ass:1})–(A\ref{ass:2}), for any test function $\varphi \in \mathcal{C}_b^2(\Theta \times \mathbb{R}^{dT}) \cap \mathcal{B}_b(\Theta \times \mathbb{R}^{dT})$, there exists a finite constant $C$ such that for any configuration $(l_{\star}, L, N_{l_{\star}}, \ldots, N_L) \in \mathbb{N}^{L-l_{\star}+3}$ satisfying $l_{\star} < L$,
\begin{equation}\label{eq:particle_telescoping_bound}
\left| \sum_{l=l_{\star}+1}^L \left[ \left( \overline{\pi}^{l,N_l}[\varphi] - \overline{\pi}^{l-1,N_l}[\varphi] \right) - \left( \overline{\pi}^l[\varphi] - \overline{\pi}^{l-1}[\varphi] \right) \right] \right| \leq C \sum_{l=l_{\star}+1}^L \frac{\sqrt{\Delta_l}}{\sqrt{N_l}}.
\end{equation}
\end{lem}

\begin{proof} 
The proof proceeds in the following two steps.

\noindent\textbf{Step 1: Posterior Representation.}
For each level $l$, recall that $P^l_{\mu^l_{k-1,\theta},k,\theta}(x_{k-1},du_k)$ denotes the extended transition kernel over the time interval $[k-1,k]$ at discretization level $l$:
\begin{equation*}
\overline{P}^l_{\mu^l_{k-1,\theta},k,\theta}(x_{k-1},du_k) = \prod_{j=1}^{\Delta_l^{-1}} Q^l_{\mu^l_{k-1+(j-1)\Delta_l,\theta},\theta}(x_{k-1+(j-1)\Delta_l}, dx_{k-1+j\Delta_l}),
\end{equation*}
where $Q_{\mu_{t-1+(k-1)\Delta_l,\theta}^l,\theta}^l(x_{t-1+(k-1)\Delta_l},dx_{t-1+k\Delta_l})$ denote the Gaussian Markov kernel on $(\mathbb{R}^d,\mathscr{B}(\mathbb{R}^d))$ in $\Delta_l$ time step of the discretized SDE \eqref{eq:sde_disc} with law $\mu_{t-1+(k-1)\Delta_l,\theta}^l\in\mathcal{P}(\mathbb{R}^d)$,
and $u_k = (x_{k-1+\Delta_l},\ldots,x_k) \in E_l = (\mathbb{R}^d)^{\Delta_l^{-1}}$ represents the intermediate states within $[k-1,k]$.

For any bounded measurable function $\psi: \Theta \times \mathbb{R}^{dT} \to \mathbb{R}$, define the path functional
\begin{equation*}
\mathcal{P}^l_{T,\theta}[\psi] := \int_{\Theta \times E_l^T} \psi(\theta, x_{1:T}) \prod_{k=1}^T \overline{P}^l_{\mu^l_{k-1,\theta},k,\theta}(x_{k-1}, du_k),
\end{equation*}
where $E_l^T = (E_l)^T$ is the path space over the full time horizon $[0,T]$ at discretization level $l$. Similarly, the particle-approximated path functional is defined as
\begin{equation*}
\mathcal{P}^{l,N_l}_{T,\theta}[\psi] := \int_{\Theta \times E_l^T} \psi(\theta, x_{1:T}) \prod_{k=1}^T \overline{P}^l_{\mu^{l,N_l}_{k-1,\theta},k,\theta}(x_{k-1}, du_k),
\end{equation*}
where $\mu^{l,N_l}_{k,\theta}$ denotes the empirical measure constructed using $N_l$ particles.

Let $\mathcal{G}_{\theta} := \prod_{k=1}^T g_{\theta}(x_k, y_k)$ denote the observation likelihood functional. The posterior representations then take the following form:
\begin{align*}
\overline{\pi}^l[\varphi] = \frac{\int_{\Theta} \mathcal{P}^l_{T,\theta}[\varphi \cdot \mathcal{G}_{\theta}] \, \nu(d\theta)}{\int_{\Theta} \mathcal{P}^l_{T,\theta}[\mathcal{G}_{\theta}] \, \nu(d\theta)} \quad\text{and}\quad
\overline{\pi}^{l,N_l}[\varphi] = \frac{\int_{\Theta} \mathbb{E}^{N_l}_{\theta}[\mathcal{P}^{l,N_l}_{T,\theta}[\varphi \cdot \mathcal{G}_{\theta}]] \, \nu(d\theta)}{\int_{\Theta} \mathbb{E}^{N_l}_{\theta}[\mathcal{P}^{l,N_l}_{T,\theta}[\mathcal{G}_{\theta}]] \, \nu(d\theta)},
\end{align*}
where $\mathbb{E}^{N_l}_{\theta}$ denotes expectation with respect to the $N_l$-particle approximation of the McKean-Vlasov laws.

\noindent\textbf{Step 2: Ratio Expansion and Coupling Construction.}
By the triangle inequality, we bound the difference in each level separately. For a fixed $l \in \{l_{\star}+1, \ldots, L\}$, the ratio difference can be expanded using the standard algebraic identity. For positive reals $A, B, C, D$ and arbitrary reals $a, b, c, d$, the following decomposition holds:
\begin{align}
\label{eq:ratio_identity_complete}
\frac{a}{A} - \frac{b}{B} - \frac{c}{C} + \frac{d}{D} 
&= \frac{1}{A}[(a-b) - (c-d)] - \frac{b}{AB}(A-B-C+D)+\mathcal{R}_{\text{rem}},
\end{align}
where 
\begin{align}
\label{eq:ratio_identity_cross}
\mathcal{R}_{\text{rem}}=&- \frac{1}{AC}(A-C)(c-d)  - \frac{1}{AB}(C-D)(b-d)\nonumber\\
&+ \frac{d}{CBD}(B-D)(C-D)  + \frac{d}{ACB}(A-C)(C-D).
\end{align}
The first two terms in the decomposition (\ref{eq:ratio_identity_complete}) constitute the dominant contributions and possess the same convergence order. We will analyze the first term in detail. The second term can be analyzed similarly, noting that the test function $\varphi$ is bounded ($\|\varphi\|_\infty < \infty$) and the term involves the coupled difference of normalizing constants. The remaining term $\mathcal{R}_{\text{rem}}$ represents higher-order error interactions. We will analyze the first term in $\mathcal{R}_{\text{rem}}$, as its other terms  can be controlled similarly.

We focus on the first term $\frac{1}{A}[(a-b) - (c-d)]$, the dominant contributor, which corresponds to 
\begin{multline}\label{eq:main_term}
\left( \int_{\Theta} \mathbb{E}^{N_l}_{\theta}[\mathcal{P}^{l,N_l}_{T,\theta}[\mathcal{G}_{\theta}]] \, \nu(d\theta) \right)^{-1} \\
\times \Bigg\{ \int_{\Theta} \mathbb{E}^{N_l}_{\theta}[\mathcal{P}^{l,N_l}_{T,\theta}[\varphi \cdot \mathcal{G}_{\theta}]] \, \nu(d\theta) - \int_{\Theta} \mathbb{E}^{N_l}_{\theta}[\mathcal{P}^{l-1,N_l}_{T,\theta}[\varphi \cdot \mathcal{G}_{\theta}]] \, \nu(d\theta) \\
- \left( \int_{\Theta} \mathcal{P}^l_{T,\theta}[\varphi \cdot \mathcal{G}_{\theta}] \, \nu(d\theta) - \int_{\Theta} \mathcal{P}^{l-1}_{T,\theta}[\varphi \cdot \mathcal{G}_{\theta}] \, \nu(d\theta) \right) \Bigg\}.
\end{multline}

To control \eqref{eq:main_term}, construct a probability space supporting four synchronized Euler-Maruyama processes driven by common Brownian increments:
\begin{itemize}
\item $X^{l,N_l}_{1:T}$: fine-level path with $N_l$-particle law approximation,
\item $X^{l-1,N_l}_{1:T}$: coarse-level path with $N_l$-particle law approximation,
\item $X^l_{1:T}$: fine-level path with exact McKean-Vlasov laws,
\item $X^{l-1}_{1:T}$: coarse-level path with exact McKean-Vlasov laws.
\end{itemize}
Denote by $\widetilde{\mathbb{E}}^{N_l}_{\theta}$ the expectation under this joint coupling. The numerator in \eqref{eq:main_term} can be rewritten as
\begin{multline}\label{eq:coupled_numerator}
\int_{\Theta} \widetilde{\mathbb{E}}^{N_l}_{\theta} \Big[ \varphi(\theta, X^{l,N_l}_{1:T}) \mathcal{G}_{\theta}(X^{l,N_l}_{1:T}) - \varphi(\theta, X^{l-1,N_l}_{1:T}) \mathcal{G}_{\theta}(X^{l-1,N_l}_{1:T}) \\
- \{ \varphi(\theta, X^l_{1:T}) \mathcal{G}_{\theta}(X^l_{1:T}) - \varphi(\theta, X^{l-1}_{1:T}) \mathcal{G}_{\theta}(X^{l-1}_{1:T}) \} \Big] \nu(d\theta).
\end{multline}

Since $\varphi \cdot \mathcal{G}_{\theta} \in \mathcal{C}_b^2 \cap \mathcal{B}_b$, Taylor expansion (see Lemma A.5 in \cite{po_mv}) gives the upper bound of the quantity in \eqref{eq:coupled_numerator} as
\begin{multline*}
C \int_{\Theta} \widetilde{\mathbb{E}}^{N_l}_{\theta} \Big[ \| (X^{l,N_l}_{1:T} - X^{l-1,N_l}_{1:T}) - (X^l_{1:T} - X^{l-1}_{1:T}) \| \\
+ \| X^l_{1:T} - X^{l-1}_{1:T} \| \cdot \left( \| X^{l,N_l}_{1:T} - X^l_{1:T} \| + \| X^{l-1,N_l}_{1:T} - X^{l-1}_{1:T} \| \right) \Big] \nu(d\theta).
\end{multline*}
Under Assumption (A\ref{ass:2}), applying Cauchy-Schwarz together with Lemma A.7 in \cite{po_mv}, the quantity in \eqref{eq:main_term} can be bounded by $O\left( \frac{\sqrt{\Delta_l}}{\sqrt{N_l}} \right)$.

We next bound the first term in $\mathcal{R}_{\text{rem}}$,
$- \frac{1}{AC}(A-C)(c-d)$, 
which corresponds to
\begin{align}
\label{eqn:1st_rem}
    -\left( \int_{\Theta} \mathbb{E}^{N_l}_{\theta}[\mathcal{P}^{l,N_l}_{T,\theta}[\mathcal{G}_{\theta}]] \, \nu(d\theta) \right)^{-1} \left( \int_{\Theta} \mathcal{P}^l_{T,\theta}[\mathcal{G}_{\theta}] \, \nu(d\theta) \right)^{-1}
\end{align}
$$\times \left| \int_{\Theta} \left(\mathbb{E}^{N_l}_{\theta}[\mathcal{P}^{l,N_l}_{T,\theta}[\mathcal{G}_{\theta}]] - \mathcal{P}^l_{T,\theta}[\mathcal{G}_{\theta}]\right) \nu(d\theta) \right|$$
$$\times \left| \int_{\Theta} \left(\mathcal{P}^l_{T,\theta}[\varphi \cdot \mathcal{G}_{\theta}] - \mathcal{P}^{l-1}_{T,\theta}[\varphi \cdot \mathcal{G}_{\theta}]\right) \nu(d\theta) \right|.$$

Under Assumption (A\ref{ass:2}), $\int_{\Theta} \mathbb{E}^{N_l}_{\theta}[\mathcal{P}^{l,N_l}_{T,\theta}[\mathcal{G}_{\theta}]] \, \nu(d\theta)$ and $\int_{\Theta} \mathcal{P}^l_{T,\theta}[\mathcal{G}_{\theta}] \, \nu(d\theta)$ are bounded away from zero. Here, $\int_{\Theta} \left(\mathbb{E}^{N_l}_{\theta}[\mathcal{P}^{l,N_l}_{T,\theta}[\mathcal{G}_{\theta}]] - \mathcal{P}^l_{T,\theta}[\mathcal{G}_{\theta}]\right) \nu(d\theta)$ representing particle approximation error, can be bounded by $O(N_l^{-1/2})$ using Lemma A.7 of \cite{po_mv}; $\int_{\Theta} \left(\mathcal{P}^l_{T,\theta}[\varphi \cdot \mathcal{G}_{\theta}] - \mathcal{P}^{l-1}_{T,\theta}[\varphi \cdot \mathcal{G}_{\theta}]\right) \nu(d\theta)$ representing time discretization weak error, can be bounded by $O(\Delta_l)$ under exact laws. Therefore, the quantity in \eqref{eqn:1st_rem} is of order 
$ O(1) \cdot O(N_l^{-1/2}) \cdot O(\Delta_l) = O\left(\frac{\Delta_l}{\sqrt{N_l}}\right).$

The other 3 terms in $\mathcal{R}_{\text{rem}}$ could be handled analogously, with each contributing $O(\Delta_l/\sqrt{N_l})$. They are dominated by the principal term's order $O(\sqrt{\Delta_l}/\sqrt{N_l})$ when $\Delta_l \to 0$. Summing over $l \in \{l_{\star}+1, \ldots, L\}$ completes the proof.
\end{proof}

\begin{lem}\label{lem:lem2}
Under Assumptions (A\ref{ass:1})–(A\ref{ass:2}), for any test function $\varphi \in \mathcal{C}_b^2(\Theta \times \mathbb{R}^{dT}) \cap \mathcal{B}_b(\Theta \times \mathbb{R}^{dT})$, there exists a finite constant $C$ such that for any $(l, N_l) \in \mathbb{N}^2$,
\begin{equation}\label{eq:base_particle_bias_bound}
\left| \overline{\pi}^{l,N_l}[\varphi] - \overline{\pi}^l[\varphi] \right| \leq \frac{C}{\sqrt{N_l}}.
\end{equation}
\end{lem}

\begin{proof}
Recall the posterior representations from the proof of Lemma \ref{lem:lem1}:
\begin{align*}
\overline{\pi}^l[\varphi] = \frac{\int_{\Theta} \mathcal{P}^l_{T,\theta}[\varphi \cdot \mathcal{G}_{\theta}] \, \nu(d\theta)}{\int_{\Theta} \mathcal{P}^l_{T,\theta}[\mathcal{G}_{\theta}] \, \nu(d\theta)} \quad\text{and}\quad
\overline{\pi}^{l,N_l}[\varphi] = \frac{\int_{\Theta} \mathbb{E}^{N_l}_{\theta}[\mathcal{P}^{l,N_l}_{T,\theta}[\varphi \cdot \mathcal{G}_{\theta}]] \, \nu(d\theta)}{\int_{\Theta} \mathbb{E}^{N_l}_{\theta}[\mathcal{P}^{l,N_l}_{T,\theta}[\mathcal{G}_{\theta}]] \, \nu(d\theta)},
\end{align*}
where $\mathcal{G}_{\theta} = \prod_{k=1}^T g_{\theta}(x_k, y_k)$ and $\mathcal{P}^l_{T,\theta}[\cdot]$, $\mathcal{P}^{l,N_l}_{T,\theta}[\cdot]$ are path functionals defined therein. The proof proceeds in the following three steps.

\noindent\textbf{Step 1: Ratio Decomposition.}
To analyze the difference $\pi^{l,N_l}[\varphi] - \pi^l[\varphi]$, we introduce simplified notation. Let
\begin{align*}
a &:= \int_{\Theta} \mathbb{E}^{N_l}_{\theta}[\mathcal{P}^{l,N_l}_{T,\theta}[\varphi \cdot \mathcal{G}_{\theta}]] \, \nu(d\theta), \quad
A := \int_{\Theta} \mathbb{E}^{N_l}_{\theta}[\mathcal{P}^{l,N_l}_{T,\theta}[\mathcal{G}_{\theta}]] \, \nu(d\theta), \\
c &:= \int_{\Theta} \mathcal{P}^l_{T,\theta}[\varphi \cdot \mathcal{G}_{\theta}] \, \nu(d\theta), \quad\text{and}\quad
C := \int_{\Theta} \mathcal{P}^l_{T,\theta}[\mathcal{G}_{\theta}] \, \nu(d\theta).
\end{align*}

Then $\pi^{l,N_l}[\varphi] = a/A$ and $\pi^l[\varphi] = c/C$. Using the algebraic identity
$$\frac{a}{A} - \frac{c}{C} = \frac{a-c}{C} + \frac{a(C-A)}{AC},$$
we decompose the difference into two terms:
\begin{equation}\label{eq:decomp_two_terms}
\pi^{l,N_l}[\varphi] - \pi^l[\varphi] = \underbrace{\frac{1}{C}(a-c)}_{\text{Term 1}} + \underbrace{\frac{a}{AC}(C-A)}_{\text{Term 2}}.
\end{equation}

\noindent\textbf{Step 2: Bounding Strategy.}
Both terms in \eqref{eq:decomp_two_terms} are controlled by the same quantity. By Assumption (A\ref{ass:2}), the denominators $C$ and $AC$ are bounded away from zero. Since $\varphi$ is bounded ($\|\varphi\|_\infty < \infty$), we have $|a| \leq C'$ for some constant $C'$. Therefore, both terms can be bounded in terms of the particle approximation error of the normalizing constant:
\begin{equation}\label{eq:key_diff}
\left| C - A \right| = \left| \int_{\Theta} \mathcal{P}^l_{T,\theta}[\mathcal{G}_{\theta}] \, \nu(d\theta) - \int_{\Theta} \mathbb{E}^{N_l}_{\theta}[\mathcal{P}^{l,N_l}_{T,\theta}[\mathcal{G}_{\theta}]] \, \nu(d\theta) \right|.
\end{equation}

\noindent\textbf{Step 3: Particle Approximation Error via Coupling.}
To bound \eqref{eq:key_diff}, we construct a probability space supporting two coupled processes:
\begin{itemize}
\item $X^{l,N_l}_{1:T}$: particle-approximated path at level $l$ with $N_l$ particles,
\item $X^l_{1:T}$: exact-law path at level $l$,
\end{itemize}
driven by common Brownian increments. Denote the joint expectation by $\widetilde{\mathbb{E}}^{N_l}_{\theta}$. Then 
\begin{equation*}
\left| C - A \right| \leq C \left| \int_{\Theta} \widetilde{\mathbb{E}}^{N_l}_{\theta} \left[ \prod_{k=1}^T g_{\theta}(X^{l,N_l}_k, y_k) - \prod_{k=1}^T g_{\theta}(X^l_k, y_k) \right] \nu(d\theta) \right|=O\left( \frac{1}{\sqrt{N_l}} \right),
\end{equation*}
since $\mathcal{G}_{\theta} \in \mathcal{C}_b^2 \cap \mathcal{B}_b$, and the bound follows from Lemma A.7 of \cite{po_mv}.

Summing up Term 1 and Term 2 results $O\left( \frac{1}{\sqrt{N_l}} \right)+O\left( \frac{1}{\sqrt{N_l}} \right)=O\left( \frac{1}{\sqrt{N_l}} \right)$, which completes the proof.
\end{proof}

\begin{lem}\label{lem:lem3}
Under Assumptions (A1)–(A2), for any test function $\varphi \in \mathcal{C}_b^2(\Theta \times \mathbb{R}^{dT}) \cap \mathcal{B}_b(\Theta \times \mathbb{R}^{dT})$, there exists a finite constant $C$ such that for any $(l, N_l) \in \mathbb{N}^2$,
\begin{align}\label{eq:coupled_variance_bound}
\int_{\Theta \times (E_l \times E_{l-1})^T} \Big(\Psi(\theta, x^l_{1:T}, \tilde{x}^{l-1}_{1:T}) \Big)^2 \check{\pi}^{l,N_l}(d(\theta, u^l_{1:T}, \tilde{u}^{l-1}_{1:T})) \leq C\Delta_l,
\end{align}
where $\Psi$ is the weighted difference functional defined as
\begin{align}
\label{eqn:Psi_SP}
\Psi(\theta, x^l_{1:T}, \tilde{x}^{l-1}_{1:T}) := \varphi(\theta, x^l_{1:T}) \prod_{k=1}^T \check{H}_{k,\theta}(x^l_k, \tilde{x}^{l-1}_k) - \varphi(\theta, \tilde{x}^{l-1}_{1:T}) \prod_{k=1}^T \check{H}_{k,\theta}(\tilde{x}^{l-1}_k, x^l_k), 
\end{align}
with $\check{H}_{k,\theta}(x, x') = \frac{g_{\theta}(x, y_k)}{H_{k,\theta}(x, x')}$ and $H_{k,\theta}(x, x') = \frac{1}{2}(g_{\theta}(x, y_k) + g_{\theta}(x', y_k))$ defined in the main context.
\end{lem}

\begin{proof}
The coupled posterior $\check{\pi}^{l,N_l}$ defined in (\ref{eq:main_tar}) is constructed from synchronized Euler-Maruyama discretizations at levels $l$ and $l-1$ driven by common Brownian increments. This coupling mechanism ensures that trajectories $(x^l_{1:T}, \tilde{x}^{l-1}_{1:T})$ exhibit strong correlation, which is essential for variance reduction in the multilevel framework. The proof proceeds in the following two steps.

\noindent\textbf{Step 1. }
Since $\varphi \in \mathcal{C}_b^2$ and $\check{H}_{k,\theta}$ is bounded under Assumption (A\ref{ass:2}), the composite function $\Psi$ admits a first-order Taylor expansion. For coupled trajectories $(x^l_{1:T}, \tilde{x}^{l-1}_{1:T})$, standard mean-value estimates yield
$$|\Psi(\theta, x^l_{1:T}, \tilde{x}^{l-1}_{1:T})| \leq C \sum_{k=1}^T \|x^l_k - \tilde{x}^{l-1}_k\|,$$
where the constant depends on $\|\varphi\|_{\mathcal{C}^1_b}$ and $\sup_{k,\theta} \|\check{H}_{k,\theta}\|_{\mathcal{C}^1_b}$.

Squaring both sides and taking expectations with respect to $\check{\pi}^{l,N_l}$:
$$\mathbb{E}_{\check{\pi}^{l,N_l}}[\Psi^2] \leq C \mathbb{E}_{\check{\pi}^{l,N_l}}\left[\left(\sum_{k=1}^T \|x^l_k - \tilde{x}^{l-1}_k\|\right)^2\right].$$
By the Cauchy-Schwarz inequality,
\begin{align*}
\mathbb{E}_{\check{\pi}^{l,N_l}}\left[\left(\sum_{k=1}^T \|x^l_k - \tilde{x}^{l-1}_k\|\right)^2\right] &\leq T \sum_{k=1}^T\mathbb{E}_{\check{\pi}^{l,N_l}}[\|x^l_k - \tilde{x}^{l-1}_k\|^2]\\&= T\sum_{k=1}^T \frac{ \mathbb{E}_{{l,N_l}}[\|x^l_k - \tilde{x}^{l-1}_k\|^2 \prod_{k=1}^T H_{k,\theta}(\tilde{x}^{l-1}_k, x^l_k)]}{ \mathbb{E}_{{l,N_l}}[ \prod_{k=1}^T H_{k,\theta}(\tilde{x}^{l-1}_k, x^l_k)]} \\
&\leq CT\sum_{k=1}^T\mathbb{E}_{{l,N_l}}[\|x^l_k - \tilde{x}^{l-1}_k\|^2].\end{align*}
where the constant $C$ follows from the bounds on $\prod_{k=1}^T H_{k,\theta}(\tilde{x}^{l-1}_k, x^l_k)$ imposed by Assumption (A\ref{ass:2}), and the expectation is taken w.r.t the joint distribution of the coupled trajectory sequences $\prod_{k=1}^T \check{P}_{\theta}^{N_l}\!\left(
(x_{k-1}^l,\widetilde{x}_{k-1}^{l-1}),
d(u_k^l,\widetilde{u}_{k}^{l-1})
\right)$.

\noindent\textbf{Step 2.}
For synchronized Euler-Maruyama discretizations at consecutive levels $l$ and $l-1$ sharing common Brownian increments, 
For synchronized Euler–Maruyama discretizations at consecutive levels $l$ and $l-1$ that share common Brownian increments, we have
$$\mathbb{E}_{l,N_l}[\|x^l_t - \tilde{x}^{l-1}_t\|^2] \leq C\Delta_l, \qquad t \in \{1, \ldots, T\},$$
which follows by applying the proof strategy used in establishing Lemma A.7 of \cite{po_mv}.
This result accounts for:
\begin{itemize}
\item The coupled Euler-Maruyama discretizations at levels $l$ and $l-1$, where both processes are driven by the same Brownian motion;
\item The particle approximation of the McKean-Vlasov laws $\mu^{l,N_l}_{t,\theta}$ and $\mu^{l-1,N_l}_{t,\theta}$ using $N_l$ particles, where both interaction kernels $\zeta_{1,\theta}$ (in drift) and $\zeta_{2,\theta}$ (in diffusion) are synchronously approximated;
\item The standard strong convergence rate $O(\Delta_l)$ for coupled discretizations, which holds under Assumption (\ref{ass:1}) and the coupling structure.
\end{itemize}

Combining the above two steps gives
$$\mathbb{E}_{\check{\pi}^{l,N_l}}[\Psi^2] \leq C \cdot T \sum_{t=1}^T \Delta_l = CT^2\Delta_l = O(\Delta_l),$$
which completes the proof.
\end{proof}

\begin{lem}\label{lem:lem4}
Under Assumptions (A\ref{ass:1})–(A\ref{ass:3}), for any test function $\varphi \in \mathcal{C}_b^2(\Theta \times \mathbb{R}^{dT}) \cap \mathcal{B}_b(\Theta \times \mathbb{R}^{dT})$, there exists a finite constant $C$ such that for any $(l_{\star}, L, N_{l_{\star}}, K_{l_{\star}}, \ldots, N_L, K_L) \in \mathbb{N}^{2(L-l_{\star})+4}$ with $l_{\star} < L$,
\begin{equation}\label{eq:bilevel_mcmc_variance}
\sum_{l=l_{\star}+1}^L \mathbb{E}\left[
\left(
\overline{\pi}^{l,N_l,K_l}[\varphi] - \overline{\pi}^{l-1,N_l,K_l}[\varphi] -
\{\overline{\pi}^{l,N_l}[\varphi] - \overline{\pi}^{l-1,N_l}[\varphi]\}
\right)^2
\right]
\leq C\sum_{l=l_{\star}+1}^L\frac{\Delta_l}{K_l+1}.
\end{equation}
\end{lem}

\begin{proof}
Fix level $l \in \{l_{\star}+1, \ldots, L\}$. The proof proceeds in the following four steps.

\noindent\textbf{Step 1: Ratio Identity and Decomposition.}
The bi-level MCMC estimator $\overline{\pi}^{l,N_l,K_l}[\varphi] - \overline{\pi}^{l-1,N_l,K_l}[\varphi]$ from equation (\ref{eq:mcmc_bl_est}) involves four ratio terms. Applying the algebraic identity from Lemma \ref{lem:lem1}, equation (\ref{eq:ratio_identity_complete}), we obtain a decomposition into six terms (two principal terms plus four remainder terms) for the difference
$$\frac{a}{A} - \frac{b}{B} - \frac{c}{C} + \frac{d}{D},$$
where
\begin{align*}
\frac{a}{A} = \overline{\pi}^{l,N_l,K_l}[\varphi],\quad
\frac{b}{B} = \overline{\pi}^{l-1,N_l,K_l}[\varphi], \quad
\frac{c}{C} = \overline{\pi}^{l,N_l}[\varphi],\quad \text{and}\quad
\frac{d}{D} = \overline{\pi}^{l-1,N_l}[\varphi].
\end{align*}
For the \textbf{principal terms}, we consider 
$$\frac{1}{A}[(a-b) - (c-d)],$$
which represents the difference-of-differences under the coupled posterior $\check{\pi}^{l,N_l}$ defined in equation (\ref{eq:main_tar}).

\noindent\textbf{Step 2: Principal Term Control via Ergodic Theory.}
Recall the weighted difference functional defined in \eqref{eqn:Psi_SP} as
$$\Psi(\theta, x^l_{1:T}, \tilde{x}^{l-1}_{1:T}) = \varphi(\theta, x^l_{1:T}) \prod_{k=1}^T \check{H}_{k,\theta}(x^l_k, \tilde{x}^{l-1}_k) - \varphi(\theta, \tilde{x}^{l-1}_{1:T}) \prod_{k=1}^T \check{H}_{k,\theta}(\tilde{x}^{l-1}_k, x^l_k).$$
The principal term corresponds to the MCMC sampling error for estimating $\mathbb{E}_{\check{\pi}^{l,N_l}}[\Psi]$ using the bi-level MCMC chain $\{V_i\}_{i=1}^{K_l}$ with state space $\Theta \times (E_l \times E_{l-1})^T$, where each $V_i = (\theta_i, x^l_{i,1:T}, \tilde{x}^{l-1}_{i,1:T})$ represents a coupled sample at iteration $i$.

By Assumption (A\ref{ass:3}), the bi-level MCMC kernel $\check{R}_l$ is uniformly ergodic with invariant measure $\check{\pi}^{l,N_l}$. Applying standard results for uniformly ergodic Markov chains (\cite{jasra_bpe_sde}, Proposition A.1) with $K_l$ iterations yields
$$\mathbb{E}\left[\left(\frac{1}{K_l}\sum_{i=1}^{K_l}\Psi(V_i) - \mathbb{E}_{\check{\pi}^{l,N_l}}[\Psi]\right)^2\right] \leq C\frac{\text{Var}_{\check{\pi}^{l,N_l}}(\Psi)}{K_l+1},$$
where $\text{Var}_{\check{\pi}^{l,N_l}}(\Psi) = \mathbb{E}_{\check{\pi}^{l,N_l}}[\Psi^2] - (\mathbb{E}_{\check{\pi}^{l,N_l}}[\Psi])^2$ is the variance under $\check{\pi}^{l,N_l}$.
By Lemma \ref{lem:lem3}, we have $\mathbb{E}_{\check{\pi}^{l,N_l}}[\Psi^2] \leq C\Delta_l$, therefore the second moment of the principal terms can be bounded by 
$C\frac{\Delta_l}{K_l+1}.$

\noindent\textbf{Step 3: Remainder Terms Control.}
Consider one representative remainder term from equation (\ref{eq:ratio_identity_complete}),
\begin{equation*}
\mathcal{R}_{\text{example}} = -\frac{1}{AC}(A-C)(c-d),
\end{equation*}
which corresponds to the explicit form
\begin{align*}
\mathcal{R}_{\text{example}} & = \left(\frac{1}{K_l+1} \sum_{k=0}^{K_l} \mathcal{P}^{l,N_l}_{T,\theta^{(k)}}[\mathcal{G}_{\theta^{(k)}}]\right)^{-1} \left(\int_{\Theta} \mathbb{E}^{N_l}_{\theta}[\mathcal{P}^{l,N_l}_{T,\theta}[\mathcal{G}_{\theta}]] \, \nu(d\theta)\right)^{-1} \\
& \quad \times \left(\frac{1}{K_l+1} \sum_{k=0}^{K_l} \mathcal{P}^{l,N_l}_{T,\theta^{(k)}}[\mathcal{G}_{\theta^{(k)}}] - \int_{\Theta} \mathbb{E}^{N_l}_{\theta}[\mathcal{P}^{l,N_l}_{T,\theta}[\mathcal{G}_{\theta}]] \, \nu(d\theta)\right) \\
& \quad \times \left(\int_{\Theta} \left(\mathbb{E}^{N_l}_{\theta}[\mathcal{P}^{l,N_l}_{T,\theta}[\varphi \cdot \mathcal{G}_{\theta}]] - \mathbb{E}^{N_l}_{\theta}[\mathcal{P}^{l-1,N_l}_{T,\theta}[\varphi \cdot \mathcal{G}_{\theta}]]\right) \nu(d\theta)\right).
\end{align*}

This term involves three factors:
\begin{itemize}
\item \textbf{Denominators}: By Assumption (A\ref{ass:2}), $(AC)^{-1}$ is in the order of $O(1)$.
\item \textbf{MCMC error} $(A-C)$: The normalizing constant difference between the MCMC empirical average (with $K_l$ iterations) and the exact expectation under the $N_l$-particle approximated laws. By the uniform ergodicity assumption, this term has MCMC sampling variance $O((K_l+1)^{-1})$.
\item \textbf{Coupled particle approximation weak error} $(c-d)$: The difference between expectations at two discretization levels $l$ and $l-1$, both approximated using $N_l$ particles. This represents the weak error of the discretization and is bounded by $O(\Delta_l)$ under the $N_l$-particle approximated laws.
\end{itemize}
Applying the Cauchy-Schwarz inequality yields:
\begin{align*}
\mathbb{E}[|\mathcal{R}_{\text{example}}|^2] &=\cdot O(1) \mathbb{E}[(A-C)^2] \cdot |c-d|^2 \\
&= O(1) \cdot O\left(\frac{1}{K_l+1}\right) \cdot O(\Delta_l^2) \\
&= O\left(\frac{\Delta_l^2}{K_l+1}\right).
\end{align*}
Since $\Delta_l \to 0$ as $l$ increases in the multilevel framework, this term is of higher order in $\Delta_l$ relative to the principal term’s bound of $O!\left(\Delta_l (K_l+1)^{-1}\right)$, and is therefore dominated by the latter.

All other remainder terms admit analogous bounds through similar analysis, each contributing at most $O(\Delta_l^2(K_l+1)^{-1})$, which are negligible compared to the principal terms.

\noindent\textbf{Step 4: Summation Over Levels.}

Combining Steps 2 and 3, and applying the Cauchy–Schwarz inequality six times (once for the principal term and once for each of the five remainder terms), we obtain for each level $l \in {l_{\star}+1, \ldots, L}$:
$$\mathbb{E}\left[\left(\overline{\pi}^{l,N_l,K_l}[\varphi] - \overline{\pi}^{l-1,N_l,K_l}[\varphi] - \{\overline{\pi}^{l,N_l}[\varphi] - \overline{\pi}^{l-1,N_l}[\varphi]\}\right)^2\right] \leq C\frac{\Delta_l}{K_l+1}.$$
Summing over $l = l_{\star}+1, \ldots, L$ completes the proof.
\end{proof}

\begin{lem}\label{lem:lem5}
Under Assumptions (A\ref{ass:1})–(A\ref{ass:3}), for any test function $\varphi \in C^2_b(\Theta \times \mathbb{R}^{dT}) \cap \mathcal{B}_b(\Theta \times \mathbb{R}^{dT})$, there exists a finite constant $C < +\infty$ such that for any $(l_{\star}, L, N_{l_{\star}}, K_{l_{\star}}, \dots, N_L, K_L) \in \mathbb{N}^{2(L-l_{\star})+4}$ with $l_{\star} < L$,
\begin{align*}
&\Bigg|\sum_{(l,q) \in \mathcal{A}_{l_{\star},L}} \mathbb{E}\bigg[ \Big\{ \overline{\pi}^{l,N_l,K_l}[\varphi] - \overline{\pi}^{l-1,N_l,K_l}[\varphi]- \left(\overline{\pi}^{l,N_l}[\varphi] - \overline{\pi}^{l-1,N_l}[\varphi]\right) \Big\}  \\
&\qquad\qquad\qquad\times \Big\{ \overline{\pi}^{q,N_q,K_q}[\varphi] - \overline{\pi}^{q-1,N_q,K_q}[\varphi] - \left(\overline{\pi}^{q,N_q}[\varphi] - \overline{\pi}^{q-1,N_q}[\varphi]\right) \Big\} \bigg]\Bigg|\\
&\leq C \sum_{(l,q) \in \mathcal{A}_{l_{\star},L}} \frac{\sqrt{\Delta_l}}{K_l+1} \cdot \frac{\sqrt{\Delta_q}}{K_q+1},
\end{align*}
where $\mathcal{A}_{l_{\star},L} = \{(l,q) \in \{l_{\star}+1, \dots, L\}^2 : l \neq q\}$.
\end{lem}

\begin{proof}
We control the cross-level covariance terms arising in the multilevel estimator. For distinct discretization levels, the independence of the corresponding bi-level MCMC chains allows these terms to factorize, enabling straightforward bounding. The proof proceeds in the following four steps.

\noindent\textbf{Step 1: Problem Setup and Independence Structure.}
For any $(l,q) \in \mathcal{A}_{l_{\star},L}$ with $l \neq q$, consider the cross-level expectation:
\begin{align*}
&\mathbb{E}\bigg[\Big\{\overline{\pi}^{l,N_l,K_l}[\varphi] - \overline{\pi}^{l-1,N_l,K_l}[\varphi] - \left(\overline{\pi}^{l,N_l}[\varphi] - \overline{\pi}^{l-1,N_l}[\varphi]\right)\Big\}\\
&\qquad\qquad \times\Big\{\overline{\pi}^{q,N_q,K_q}[\varphi] - \overline{\pi}^{q-1,N_q,K_q}[\varphi] - \left(\overline{\pi}^{q,N_q}[\varphi] - \overline{\pi}^{q-1,N_q}[\varphi]\right)\Big\}\bigg].
\end{align*}
By Algorithm \ref{MLPMCMC_mckean_vlasov_part1}, the bi-level MCMC chains at levels $l$ and $q$ are run independently, using independent randomness (Brownian motions and resampling indices). By independence, the expectation factors as:
\begin{align}
\mathbb{E}\left[\text{Cross-term}_{l,q}\right]
&= \mathbb{E}\left[\overline{\pi}^{l,N_l,K_l}[\varphi] - \overline{\pi}^{l-1,N_l,K_l}[\varphi] - \left(\overline{\pi}^{l,N_l}[\varphi] - \overline{\pi}^{l-1,N_l}[\varphi]\right)\right]\label{eq:factorization_A10}\\
&\quad\times \mathbb{E}\left[\overline{\pi}^{q,N_q,K_q}[\varphi] - \overline{\pi}^{q-1,N_q,K_q}[\varphi] - \left(\overline{\pi}^{q,N_q}[\varphi] - \overline{\pi}^{q-1,N_q}[\varphi]\right)\right].\nonumber
\end{align}

\noindent\textbf{Step 2: Bounding Individual Bias Terms.}
Define the bias at level $l$ as:
\begin{equation*}
B_l := \mathbb{E}\left[\overline{\pi}^{l,N_l,K_l}[\varphi] - \overline{\pi}^{l-1,N_l,K_l}[\varphi]\right] - \left(\overline{\pi}^{l,N_l}[\varphi] - \overline{\pi}^{l-1,N_l}[\varphi]\right).
\end{equation*}
Following the ratio decomposition approach in the proof of Lemma \ref{lem:lem1}, we adopt notation consistent with equation (\ref{eq:ratio_identity_complete}). For level $l$, define
\begin{itemize}
\item $A := \frac{1}{K_l}\sum_{i=1}^{K_l} \prod_{k=1}^T \check{H}{k,\theta^{(i)}}(x^{l,i}k, \tilde{x}^{,l-1,i}k)$, the fine-level MCMC normalizing constant;
\item $a := \frac{1}{K_l}\sum{i=1}^{K_l} \varphi(\theta^{(i)}, x^{l,i}{1:T}) \prod{k=1}^T \check{H}_{k,\theta^{(i)}}(x^{l,i}_k, \tilde{x}^{,l-1,i}_k)$, the fine-level MCMC weighted expectation;
\item $B, b$: the corresponding quantities for the coarse-level ($l-1$) MCMC;
\item $C, c$: the corresponding fine-level exact expectations under $\check{\pi}^{l,N_l}$;
\item $D, d$: the corresponding coarse-level exact expectations under $\check{\pi}^{l,N_l}$.
\end{itemize}
The bi-level difference can be expressed as:
\begin{equation*}
\overline{\pi}^{l,N_l,K_l}[\varphi] - \overline{\pi}^{l-1,N_l,K_l}[\varphi] = \frac{a}{A} - \frac{b}{B}.
\end{equation*}

Applying the ratio identity from Lemma \ref{lem:lem1}, equation (\ref{eq:ratio_identity_complete}), and noting that $\frac{c}{C} = \overline{\pi}^{l,N_l}[\varphi]$ and $\frac{d}{D} = \overline{\pi}^{l-1,N_l}[\varphi]$, we obtain a decomposition of the difference
\begin{equation*}
\frac{a}{A} - \frac{b}{B} - \frac{c}{C} + \frac{d}{D}
\end{equation*}
into six terms—two principal terms and four remainder terms.

\textit{Principal Terms Analysis:} From equation (\ref{eq:ratio_identity_complete}), we take the principal term 
\begin{align*}
T_1 = \frac{1}{A}[(a-b) - (c-d)]    
\end{align*} as an example.
We rewrite this by subtracting and adding $\frac{1}{C}[(a-b) - (c-d)]$:
\begin{equation*}
T_1 = \left(\frac{1}{A} - \frac{1}{C}\right)[(a-b) - (c-d)] + \frac{1}{C}[(a-b) - (c-d)].
\end{equation*}
Taking expectations and noting that under stationarity $\mathbb{E}[(a-b) - (c-d)] = 0$, the second term vanishes:
\begin{equation*}
\mathbb{E}[T_1] = \mathbb{E}\left[\left(\frac{1}{A} - \frac{1}{C}\right)[(a-b) - (c-d)]\right].
\end{equation*}
Applying the Cauchy-Schwarz inequality yields
\begin{equation*}
|\mathbb{E}[T_1]| \leq \sqrt{\mathbb{E}\left[\left(\frac{1}{A} - \frac{1}{C}\right)^2\right]} 
\cdot \sqrt{\mathbb{E}[((a-b) - (c-d))^2]}.
\end{equation*}

\textit{First factor:} The term $\left(\frac{1}{A} - \frac{1}{C}\right)$ represents the MCMC error 
in estimating the normalizing constant. By standard results for uniformly ergodic Markov chains 
(Assumption A\ref{ass:3}) and uniform bounds on normalizing constants $C'$ (Assumption A\ref{ass:2}),
\begin{equation*}
\sqrt{\mathbb{E}\left[\left(\frac{1}{A} - \frac{1}{C}\right)^2\right]} \leq \frac{C'}{\sqrt{K_l+1}}.
\end{equation*}

\textit{Second factor:} This term represents the variance of the coupled bi-level difference under 
MCMC sampling. By Lemma \ref{lem:lem4} (variance bound for coupled trajectories) combined with MCMC mixing 
properties from Assumption A\ref{ass:3}, we have
\begin{equation*}
\mathbb{E}[((a-b) - (c-d))^2] \leq C\frac{\Delta_l}{K_l+1}.
\end{equation*}
Therefore,
\begin{equation*}
\sqrt{\mathbb{E}[((a-b) - (c-d))^2]} \leq C\frac{\sqrt{\Delta_l}}{\sqrt{K_l+1}}.
\end{equation*}
Combining both factors gives
\begin{equation*}
|\mathbb{E}[T_1]| \leq C' \cdot \frac{1}{\sqrt{K_l+1}} \cdot \frac{\sqrt{\Delta_l}}{\sqrt{K_l+1}} 
= C'\frac{\sqrt{\Delta_l}}{K_l+1}.
\end{equation*}

\textit{Remainder Terms Analysis:} We consider the remainder term
\begin{equation*}
R_3 = -\frac{1}{AC}(A-C)(c-d).
\end{equation*}

We rewrite this by subtracting and adding $\frac{1}{C^2}(A-C)(c-d)$:
\begin{equation*}
R_3 = \left(\frac{1}{AC} - \frac{1}{C^2}\right)(A-C)(c-d) + \frac{1}{C^2}(A-C)(c-d).
\end{equation*}

Taking expectations and noting that $c, d,$ and $C$ are deterministic with respect to the MCMC randomness and $\mathbb{E}[A-C] = 0$ under stationarity, the second term vanishes:
\begin{equation*}
\mathbb{E}\left[\frac{1}{C^2}(A-C)(c-d)\right] = \frac{c-d}{C^2}\mathbb{E}[A-C] = 0.
\end{equation*}

Thus, the expectation is determined by the first term. Using the identity $\frac{1}{AC} - \frac{1}{C^2} = \frac{C-A}{AC^2}$:
\begin{equation*}
\mathbb{E}[R_3] = \mathbb{E}\left[-\frac{1}{AC^2}(A-C)^2(c-d)\right].
\end{equation*}

Taking absolute values and noting that normalizing constants are bounded away from zero (Assumption A\ref{ass:2}):
\begin{equation*}
|\mathbb{E}[R_3]| \leq C |c-d| \mathbb{E}[(A-C)^2].
\end{equation*}

\textit{Weak error factor:} The term $|c-d|$ represents the weak error of the time discretization. By standard error analysis for Euler-Maruyama schemes,
\begin{equation*}
|c-d| \leq C \Delta_l.
\end{equation*}

\textit{Variance factor:} The term $\mathbb{E}[(A-C)^2]$ represents the variance of the MCMC estimator. By uniform ergodicity (Assumption A\ref{ass:3}) and Corollary 2.1 in \cite{roberts1997geometric},
\begin{equation*}
\mathbb{E}[(A-C)^2] \leq \frac{C'}{K_l+1}.
\end{equation*}
Combining these bounds:
\begin{equation*}
|\mathbb{E}[R_3]| \leq C \cdot \Delta_l \cdot \frac{1}{K_l+1}.
\end{equation*}
Comparing this to the principal term bound $O(\frac{\sqrt{\Delta_l}}{K_l+1})$, and noting that $\Delta_l \ll \sqrt{\Delta_l}$ as $\Delta_l \to 0$, the remainder term is strictly dominated.
The remaining terms in $\mathcal{R}_{\text{rem}}$ can be analyzed similarly.

\textit{Combined Bound:} Summing the principal term and all remainder terms gives
\begin{equation*}\label{eq:bias_bound_l_A10}
|B_l| \leq  O\left(\frac{\sqrt{\Delta_l}}{K_l+1}\right).
\end{equation*}
Similarly, for level $q$,
\begin{equation*}\label{eq:bias_bound_q_A10}
|B_q| \leq C\frac{\sqrt{\Delta_q}}{K_q + 1}.
\end{equation*}

\noindent\textbf{Step 3: Combining Bounds.}
\begin{equation*}
\left|\mathbb{E}\left[\text{Cross-term}_{l,q}\right]\right| 
= |B_l| \cdot |B_q|
\leq C\frac{\sqrt{\Delta_l}}{K_l + 1} \cdot C\frac{\sqrt{\Delta_q}}{K_q + 1}
= C\frac{\sqrt{\Delta_l\Delta_q}}{(K_l+1)(K_q+1)}.
\end{equation*}
Summing over all distinct pairs $(l,q) \in \mathcal{A}_{l_\star,L}$,
\begin{equation*}
\sum_{(l,q) \in \mathcal{A}_{l_\star,L}} \left|\mathbb{E}\left[\text{Cross-term}_{l,q}\right]\right|
\leq C\sum_{(l,q) \in \mathcal{A}_{l_\star,L}} \frac{\sqrt{\Delta_l}}{K_l+1} \cdot \frac{\sqrt{\Delta_q}}{K_q+1},
\end{equation*}
which establishes the stated bound.
\end{proof}

\end{document}